\documentclass
  [ acmsmall
  , pdftex
  , dvipsnames
  , nonacm
  , screen
  ]{acmart}

\authorsaddresses{}
\setcopyright{none}
\usepackage{graphics}
\graphicspath{{_build/}{fig/}}

\makeatletter
\def\@ACM@badge@width{8mm}
\makeatother
\acmBadgeL{egg.png}

\title{\egg: Fast and Extensible Equality Saturation}

\usepackage{xspace}
\usepackage{xargs}
\usepackage[
  colorinlistoftodos,
  prependcaption,
]{todonotes}
\usepackage{fontenc}

\usepackage{subcaption}

\usepackage{hyperref}

\usepackage{listings}
\lstdefinelanguage{Rust}{
  sensitive,
  morecomment=[l]{//},
  morecomment=[s]{/*}{*/},
  moredelim=[s][{\itshape\color[rgb]{0,0,0.75}}]{\#[}{]},
  morestring=[b]{"},
  alsodigit={},
  alsoother={},
  alsoletter={!},
  otherkeywords={=>},
  morekeywords={break, continue, else, for, if, in, loop, match, return, while},
  morekeywords={as, const, let, move, mut, ref, static, unsafe},
  morekeywords={dyn, enum, fn, impl, Self, self, struct, trait, type, use, where},
  morekeywords={crate, extern, mod, pub, super},
  morekeywords={abstract, alignof, become, box, do, final, macro,
    offsetof, override, priv, proc, pure, sizeof, typeof, unsized, virtual, yield},
  morekeywords=[2]{Send},
  morekeywords=[3]{bool, char, f32, f64, i8, i16, i32, i64, isize, str, u8, u16, u32, u64, unit, usize, i128, u128},
}%

\author{Max Willsey}
\affiliation{\institution{University of Washington, Seattle} \country{USA}}
\author{Chandrakana Nandi}
\affiliation{\institution{University of Washington, Seattle} \country{USA}}
\author{Yisu Remy Wang}
\affiliation{\institution{University of Washington, Seattle} \country{USA}}
\author{Oliver Flatt}
\affiliation{\institution{University of Utah, Salt Lake City} \country{USA}}
\author{Zachary Tatlock}
\affiliation{\institution{University of Washington, Seattle} \country{USA}}
\author{Pavel Panchekha}
\affiliation{\institution{University of Utah, Salt Lake City} \country{USA}}

\newcommand{\missing}[1]{{\color{red}\bfseries [TODO]}}

\newcommand{\egg}{\texorpdfstring{\MakeLowercase{\texttt{egg}}}{\texttt{egg}}\xspace}
\newcommand{\Egg}{\texorpdfstring{\MakeLowercase{\texttt{egg}}}{\texttt{egg}}\xspace}
\newcommand{\egraphs}{\mbox{e-graphs}\xspace}
\newcommand{\egraph}{\mbox{e-graph}\xspace}
\newcommand{\Egraph}{\mbox{E-graph}\xspace}
\newcommand{\Egraphs}{\mbox{E-graphs}\xspace}
\newcommand{\eclass}{\mbox{e-class}\xspace}
\newcommand{\Eclass}{\mbox{E-class}\xspace}
\newcommand{\enode}{\mbox{e-node}\xspace}
\newcommand{\eclasses}{\mbox{e-classes}\xspace}
\newcommand{\enodes}{\mbox{e-nodes}\xspace}

\newcommand{\sz}{Szalinski\xspace}
\newcommand{\find}{\texttt{find}\xspace}

\newcommand{\equivid}{\equiv_{\sf id}}
\newcommand{\equivnode}{\equiv_{\sf node}}
\newcommand{\equivterm}{\equiv_{\sf term}}

\newcommand{\defTodo}[2]{%
  \expandafter\newcommand\csname #1\endcsname[1]{%
    \todo[linecolor=#2,backgroundcolor=#2!25,bordercolor=#2,inline,size=\tiny]{\textbf{#1}: ##1}}}

\newcommand{\defTODO}[2]{%
  \expandafter\newcommand\csname #1\endcsname[1]{%
    \todo[linecolor=#2,backgroundcolor=#2!25,bordercolor=#2,inline,size=\tiny,caption={\textbf{(#1 LONG TODO)}}]{##1}}}

\newcommand{\LeoColor}{Goldenrod}
\newcommand{\MaxColor}{RoyalBlue}
\newcommand{\RemyColor}{Red}
\newcommand{\OliverColor}{OliveGreen}
\newcommand{\ChandraColor}{Magenta}
\newcommand{\PavelColor}{Cerulean}
\newcommand{\ZachColor}{Plum}
\newcommand{\BenColor}{Orchid}
\newcommand{\JamesColor}{Salmon}

\defTodo{Leo}{\LeoColor}
\defTodo{Max}{\MaxColor}
\defTodo{Remy}{\RemyColor}
\defTodo{Oliver}{\OliverColor}
\defTodo{Chandra}{\ChandraColor}
\defTodo{Pavel}{\PavelColor}
\defTodo{Zach}{\ZachColor}
\defTodo{Ben}{\BenColor}
\defTodo{James}{\JamesColor}

\defTODO{LEO}{\LeoColor}
\defTODO{MAX}{\MaxColor}
\defTODO{REMY}{\RemyColor}
\defTODO{OLIVER}{\OliverColor}
\defTODO{CHANDRA}{\ChandraColor}
\defTODO{PAVEL}{\PavelColor}
\defTODO{ZACH}{\ZachColor}
\defTODO{BEN}{\BenColor}
\defTODO{JAMES}{\JamesColor}

\newcommand{\CongrSpeedup}{\ensuremath{88\times}\xspace}
\newcommand{\TotalSpeedup}{\ensuremath{21\times}\xspace}
\newcommand{\RepairsR}{\ensuremath{0.98}\xspace}
\newcommand{\RepairsP}{3.6e-47\xspace}
\newcommand{\nEggTests}{32\xspace}
\newcommand{\nEggTimeouts}{8\xspace}

\begin{abstract}
  An \egraph efficiently represents
  a congruence relation over many expressions.
Although they were originally developed in the late 1970s
  for use in automated theorem provers,
  a more recent technique known as \textit{equality saturation}
  repurposes \egraphs to implement state-of-the-art,
  rewrite-driven compiler optimizations and program synthesizers.
However, \egraphs remain unspecialized for this newer use case.
Equality saturation workloads exhibit distinct characteristics and
  often require ad~hoc \egraph extensions to
  incorporate transformations beyond purely syntactic rewrites.

This work contributes two techniques that make
  \egraphs fast and extensible, specializing them to equality saturation.
A new amortized invariant restoration technique called \textit{rebuilding}
  takes advantage of equality saturation's distinct workload,
  providing asymptotic speedups over current techniques in practice.
A general mechanism called \textit{\eclass analyses}
  integrates domain-specific analyses into the \egraph,
  reducing the need for ad hoc manipulation.

We implemented these techniques in
  a new open-source library called \egg.
Our case studies on
  three previously published applications of equality saturation
  highlight how \egg's performance and flexibility
  enable state-of-the-art results across diverse domains.

\end{abstract}

\begin{document}
\maketitle
\renewcommand{\shortauthors}{Willsey et al.}

\makeatletter
\providecommand\@dotsep{5}
\makeatother

%
%
%
%

\section{Introduction}
\label{sec:intro}

Equality graphs (\egraphs) were originally developed to
  efficiently represent congruence relations
  in automated theorem provers (ATPs).
At a high level, \egraphs~\cite{nelson, pp-congr}
  extend union-find~\cite{unionfind} to compactly represent
  equivalence classes of expressions while
  maintaining a key invariant:
  the equivalence relation is closed under congruence.\footnote{
    Intuitively, congruence simply means
    that $a \equiv b$ implies $f(a) \equiv f(b)$.}

Over the past decade, several projects have repurposed \egraphs
  to implement state-of-the-art, rewrite-driven
  compiler optimizations and program synthesizers
  using a technique known as \textit{equality saturation}~\cite{
    denali, eqsat, eqsat-llvm, szalinski, yogo-pldi20, spores, herbie}.
Given an input program $p$,
  equality saturation constructs an \egraph $E$ that
  represents a large set of programs equivalent to $p$,
  and then extracts the ``best'' program from $E$.
The \egraph is grown by repeatedly applying
  pattern-based rewrites. 
Critically, these rewrites only add information to the \egraph,
  eliminating the need for careful ordering.
Upon reaching a fixed point (\textit{saturation}),
  $E$ will represent \textit{all equivalent ways} to
  express $p$ with respect to the given rewrites.
After saturation (or timeout),
  a final \textit{extraction} procedure
  analyzes $E$ and selects the
  optimal program according to
  a user-provided cost function.

Ideally, a user could simply provide
  a language grammar and rewrites,
  and equality saturation would produce a effective optimizer.
Two challenges block this ideal.
First, maintaining congruence can become expensive as $E$ grows.
In part, this is because \egraphs from the conventional ATP setting
  remain unspecialized to the distinct \textit{equality saturation workload}.
Second, many applications critically depend on
  \textit{domain-specific analyses}, but
  integrating them requires ad~hoc extensions to the \egraph.
The lack of a general extension mechanism
  has forced researchers to re-implement
  equality saturation from scratch several times~\cite{herbie, eqsat, wu_siga19}.
These challenges limit equality saturation's practicality.

\textit{Equality Saturation Workload. $\,$}
ATPs frequently query and modify \egraphs and
  additionally require \textit{backtracking} to
  undo modifications (e.g., in  DPLL(T)~\cite{dpll}).
These requirements force conventional \egraph designs
  to maintain the congruence invariant after every operation.
In contrast,
  the equality saturation workload does not require backtracking and
  can be factored into distinct phases of
  (1) querying the \egraph to simultaneously find all rewrite matches and
  (2) modifying the \egraph to merge in equivalences for all matched terms.

We present a new amortized algorithm
  called \textit{rebuilding} that defers \egraph invariant maintenance
  to equality saturation phase boundaries without compromising soundness.
Empirically, rebuilding provides asymptotic speedups
  over conventional approaches.

\textit{Domain-specific Analyses. $\,$}
Equality saturation is primarily driven by syntactic rewriting,
  but many applications require additional interpreted reasoning
  to bring domain knowledge into the \egraph.
Past implementations have resorted to
  ad~hoc \egraph manipulations
  to integrate what would otherwise be
  simple program analyses like constant folding.

To flexibly incorporate such reasoning,
  we introduce a new, general mechanism called \textit{\eclass analyses}.
An \eclass analysis annotates each \eclass
  (an equivalence class of terms)
  with facts drawn from a semilattice domain.
As the \egraph grows,
  facts are introduced, propagated, and joined
  to satisfy the \textit{\eclass analysis invariant},
  which relates analysis facts to the terms represented in the \egraph.
Rewrites cooperate with \eclass analyses by
  depending on analysis facts and
  adding equivalences that in turn
  establish additional facts.
Our case studies and examples
  (Sections \ref{sec:impl} and \ref{sec:case-studies})
  demonstrate \eclass analyses like
  constant folding and free variable analysis
  which required bespoke customization in
  previous equality saturation implementations.

\textit{\Egg. $\,$}
We implement rebuilding and \eclass analyses in
  an open-source\footnote{
    web: \url{https://egraphs-good.github.io},
    source: \url{https://github.com/egraphs-good/egg},
    documentation: \url{https://docs.rs/egg}
  }
  library called \egg (\textbf{e}-\textbf{g}raphs \textbf{g}ood).
\Egg specifically targets equality saturation,
  taking advantage of its workload characteristics and
  supporting easy extension mechanisms to
  provide \egraphs specialized for
  program synthesis and optimization.
\Egg also addresses more prosaic challenges,
  e.g., parameterizing over user-defined
  languages, rewrites, and cost functions
  while still providing an optimized implementation.
Our case studies demonstrate how \egg's features
  constitute a general, reusable \egraph library that can
  support equality saturation across diverse domains.

In summary, the contributions of this paper include:

\begin{itemize}


\item Rebuilding (\autoref{sec:rebuilding}),
  a technique that restores key correctness and performance invariants
  only at select points in the equality saturation algorithm.
  Our evaluation demonstrates that rebuilding is faster than
  existing techniques in practice.

\item \Eclass analysis (\autoref{sec:extensions}),
  a technique for integrating domain-specific analyses
  that cannot be expressed as purely syntactic rewrites.
  The \eclass analysis invariant provides the guarantees
  that enable cooperation between rewrites and analyses.



\item A fast, extensible implementation of
  \egraphs in a library dubbed \egg (\autoref{sec:impl}).

\item Case studies of real-world, published tools that use \egg
    for deductive synthesis and program optimization across domains such as
    floating point accuracy,
    linear algebra optimization,
    and CAD program synthesis
    (\autoref{sec:case-studies}).
    Where previous implementations existed,
      \egg is orders of magnitude faster and offers more features.
\end{itemize}

\section{Background}
\label{sec:background}

\egg builds on \egraphs and equality saturation.
  This section describes those techniques and
  presents the challenges that \egg addresses.



\subsection{\Egraphs}
\label{sec:egraphs}

An \textit{\egraph} is a data structure that stores a set of terms and a
  congruence relation over those terms.
Originally developed for and still used in the
  heart of theorem provers~\cite{nelson, simplify, z3},
  \egraphs have also been used to power a program optimization technique
  called \textit{equality saturation}~%
  \cite{denali, eqsat, eqsat-llvm, szalinski, yogo-pldi20, spores, herbie}.

\subsubsection{Definitions}

\begin{figure}
  \centering
  \begin{align*}
     \text{function symbols} \quad & f,g                                   \\[-0.2em]
     \text{\eclass ids} \quad & a,b & \text{opaque identifiers}            \\[-0.2em]
     \text{terms}     \quad & t  ::= f \mid f(t_1, \ldots, t_m) & m \geq 1 \\[-0.2em]
     \text{\enodes}   \quad & n  ::= f \mid f(a_1, \ldots, a_m) & m \geq 1 \\[-0.2em]
     \text{\eclasses} \quad & c  ::= \{ n_1, \ldots, n_m \}     & m \geq 1
  \end{align*}
  \caption{
    Syntax and metavariables for the components of an \egraph.
    Function symbols may stand alone as constant \enodes and terms.
    An \eclass id is an opaque identifier that can be compared for equality with $=$.
  }
  \label{fig:syntax}
\end{figure}

Intuitively,
  an \egraph is a set of equivalence classes (\textit{\eclasses}).
Each \eclass is a set of \textit{\enodes} representing equivalent terms from a given language,
  and an \enode is a function symbol paired with a list of children \eclasses.
More precisely:

\begin{definition}[Definition of an \Egraph]
  \label{def:egraph}

  Given the definitions and syntax in \autoref{fig:syntax},
  an \textit{\egraph} is a tuple $(U, M, H)$ where:
  \begin{itemize}
    \item
    A union-find data structure~\cite{unionfind} $U$
      stores an equivalence relation (denoted with $\equivid$)
      over \eclass ids.

    \item
    The \textit{\eclass map} $M$ maps \eclass ids to \eclasses.
    All equivalent \eclass ids map to the same \eclass, i.e.,
      $a \equivid b$ iff $M[a]$ is the same set as $M[b]$.
    An \eclass id $a$ is said to \textit{refer to} the \eclass $M[\find(a)]$.

    \item The \textit{hashcons}\footnote{
      We use the term \textit{hashcons} to evoke the memoization technique,
      since both avoid creating new duplicates of existing objects.
    }
    $H$ is a map from \enodes to \eclass ids.
  \end{itemize}


  Note that an e-class has an identity
   (its canonical \eclass id),
   but an \enode does not.\footnote{
    Our definition of an \egraph reflects \egg's design
      and therefore differs with some other \egraph definitions and implementations.
    In particular, making e-classes but not e-nodes identifiable is unique to
      our definition.
  }
  We use \eclass id $a$ and the \eclass $M[\find(a)]$ synonymously when clear from the context.

\end{definition}

\begin{definition}[Canonicalization]
    An \egraph's union-find $U$ provides a \find operation that canonicalizes \eclass ids
      such that ${\find(U, a) = \find(U, b)}$ iff ${a \equivid b}$.
    We omit the first argument of \find where clear from context.
    \begin{itemize}
      \item An \eclass id $a$ is canonical iff $\find(a) = a$.
      \item \raggedright
            An \enode $n$ is canonical iff $n = \texttt{canonicalize}(n)$,
            where ${\texttt{canonicalize}(f(a_{1}, a_{2}, ...)) = f(\find(a_{1}), \find(a_{2}), ...)}$.
    \end{itemize}
\end{definition}

\begin{definition}[Representation of Terms]
  An \egraph, \eclass, or \enode is said to \textit{represent} a term $t$ if $t$ can be
    ``found'' within it. Representation is defined recursively:
  \begin{itemize}
    \item An \egraph represents a term if any of its \eclasses do.
    \item An \eclass $c$ represents a term if any \enode $n \in c$ does.
    \item An \enode $f(a_{1}, a_{2}, ...)$ represents a term $f(t_{1}, t_{2}, ...)$
          if they have the same function symbol $f$
          and \eclass $M[a_{i}]$ represents term $t_{i}$.
  \end{itemize}

  When each \eclass is a singleton (containing only one \enode),
    an \egraph is essentially a term graph with sharing.
  \autoref{fig:egraph-rewrite1} shows an \egraph that represents the
    expression $(a \times 2) / 2$.
\end{definition}

\begin{definition}[Equivalence]
  An \egraph defines three equivalence relations.
  \begin{itemize}
    \item Over \eclass ids: $a \equivid b$ iff $\find(a) = \find(b)$.
    \item Over \enodes: $n_{1} \equivnode n_{2}$ iff \enodes $n_{1}, n_{2}$ are in the same \eclass, 
          i.e., $\exists a.\ n_{1}, n_{2} \in M[a]$.
    \item Over terms: $t_{1} \equivterm t_{2}$ iff terms $t_{1}, t_{2}$ are represented in the same \eclass.
  \end{itemize}

  We use $\equiv$ without the subscript when the relation is clear from context.
\end{definition}

\begin{definition}[Congruence]
  For a given \egraph, let $\cong$ denote a congruence relation over \enodes such that
  ${f(a_{1}, a_{2}, ...) \cong f(b_{1}, b_{2}, ...)}$ iff $a_{i} \equivid b_{i}$.
  Let $\cong^{*}$ denote the congruence closure of $\equivnode$,
   i.e., the smallest superset of $\equivnode$ that is also a superset of $\cong$.
  Note that there may be two \enodes such that
    $n_{1} \cong^{*} n_{2}$ but
    $n_{1} \not\cong n_{2}$ and
    $n_{1} \not\equivnode n_{2}$.
  The relation $\cong$ only represents a single step of congruence;
  more than one step may be required to compute the congruence closure.
\end{definition}

\subsubsection{\Egraph Invariants}
\label{sec:invariants}

The \egraph must maintain invariants in order to
  correctly and efficiently implement the operations given in \autoref{sec:interface}.
This section only defines the invariants,
  discussion of how they are maintained is deferred to \autoref{sec:rebuild}.
These are collectively referred to as the \textit{e-graph invariants}.

\begin{definition}[The Congruence Invariant]
  \label{def:cong-inv}
  The equivalence relation over \enodes must be closed over congruence,
    i.e., $(\equivnode) = (\cong^{*})$.
  The \egraph must ensure that congruent \enodes are in the same \eclass.
  Since identical \enodes are trivially congruent,
   this implies that an \enode must be uniquely contained in a single \eclass.
\end{definition}

\begin{definition}[The Hashcons Invariant]
  \label{def:hash-inv}
  The hashcons $H$ must map all canonical \enodes to their \eclass ids.
  In other words:
  $$ \enode\ n \in M[a] \iff H[\texttt{canonicalize}(n)] = \find(a) $$

  If the hashcons invariant holds, then a procedure $\texttt{lookup}$
    can quickly find which \eclass (if any) has an \enode congruent to a given \enode $n$:
  $\texttt{lookup}(n) = H[\texttt{canonicalize}(n)]$.
\end{definition}

\begin{figure}
  \begin{subfigure}[t]{0.175\linewidth}
    \centering
    \includegraphics[height=30mm]{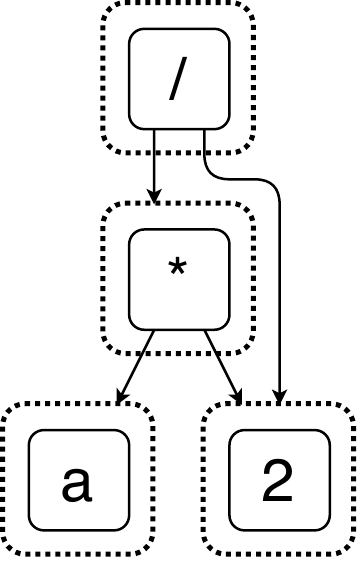}
    \caption{Initial \egraph contains ${(a \times 2) / 2}$.}
    \label{fig:egraph-rewrite1}
  \end{subfigure}
  \hfill
  \begin{subfigure}[t]{0.23\linewidth}
    \centering
    \includegraphics[height=30mm]{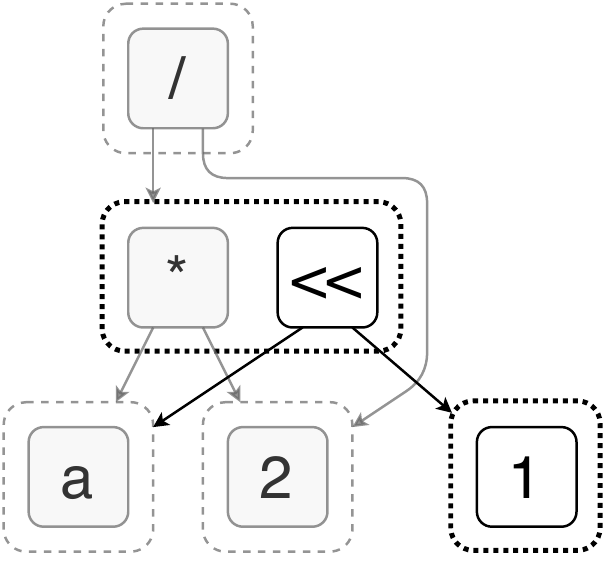}
    \caption{
      After applying rewrite ${x \times 2 \to x \ll 1}$.
    }
    \label{fig:egraph-rewrite2}
  \end{subfigure}
  \hfill
  \begin{subfigure}[t]{0.23\linewidth}
    \centering
    \includegraphics[height=30mm]{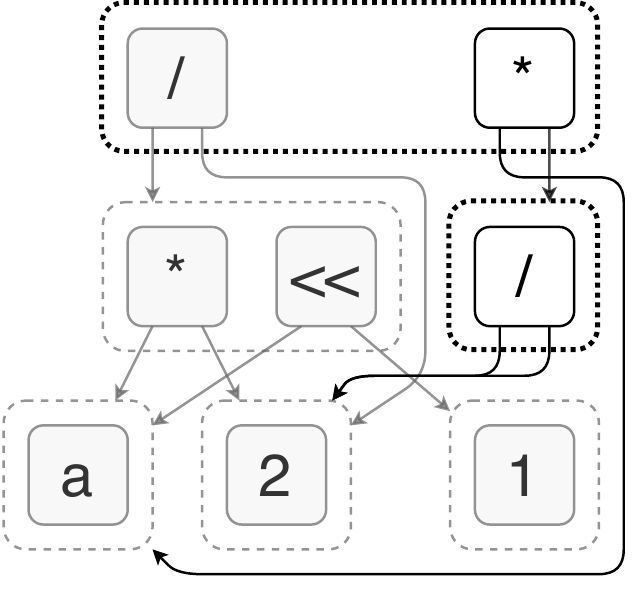}
    \caption{
      After applying rewrite ${(x \times y) / z \to x \times (y / z)}$.
    }
    \label{fig:egraph-rewrite3}
  \end{subfigure}
  \hfill
  \begin{subfigure}[t]{0.24\linewidth}
    \centering
    \includegraphics[height=30mm]{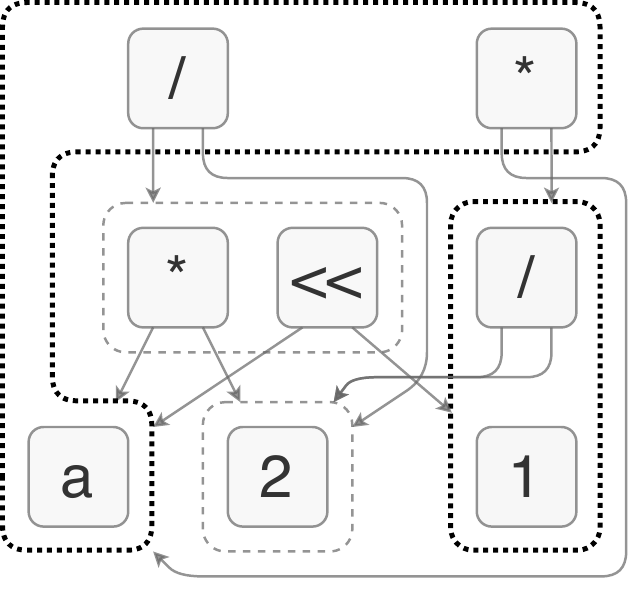}
    \caption{
      After applying rewrites ${x / x \to 1}$ and ${1 \times x \to x}$.
    }
    \label{fig:egraph-rewrite4}
  \end{subfigure}
  \caption{
    An \egraph consists of \eclasses (dashed boxes) containing
      equivalent \enodes (solid boxes).
    Edges connect \enodes to their child \eclasses.
    Additions and modifications are emphasized in black.
    Applying rewrites to an \egraph adds new \enodes and edges,
      but nothing is removed.
    Expressions added by rewrites are merged with the matched \eclass.
    In \autoref{fig:egraph-rewrite4}, the rewrites do not add any new nodes,
      only merge \eclasses.
    The resulting \egraph has a cycle,
      representing infinitely many expressions:
      $a$, $a \times 1$, $a \times 1 \times 1$, and so on.
  }
  \label{fig:egraph-rewrite}
\end{figure}

\subsubsection{Interface and Rewriting}
\label{sec:interface}

\Egraphs bear many similarities to the classic union-find data
  structure that they employ internally,
  and they inherit much of the terminology.
\Egraphs provide two main low-level mutating operations:
\begin{itemize}
    \item \texttt{add} takes an \enode $n$ and:
    \begin{itemize}
        \item if $\texttt{lookup}(n) = a$, return $a$;
        \item if $\texttt{lookup}(n) = \emptyset$,
              then set $M[a] = \{ n \}$ and return the id $a$.
    \end{itemize}
    \item \texttt{merge} (sometimes called \texttt{assert} or \texttt{union})
    takes two \eclass ids $a$ and $b$,
    unions them in the union-find $U$,
    and combines the \eclasses by setting both $M[a]$ and $M[b]$ to $M[a] \cup M[b]$.
\end{itemize}

Both of these operations must take additional steps to maintain the congruence
  invariant.
Invariant maintenance is discussed in \autoref{sec:rebuilding}.

\Egraphs also offers operations for querying the data structure.
\begin{itemize}
    \item \texttt{find} canonicalizes \eclass ids using the union-find $U$ as described in definition \ref{def:egraph}.
    \item \texttt{ematch} performs the
          \textit{e-matching}~\cite{simplify, ematching}
          procedure for finding patterns in the \egraph.
          \texttt{ematch} takes a pattern term $p$ with variable placeholders
          and returns a list of tuples $(\sigma, c)$ where $\sigma$ is a substitution of
          variables to \eclass ids such that $p[\sigma]$ is represented in \eclass $c$.
\end{itemize}
These can be composed to perform rewriting over the
  \egraph.
To apply a rewrite $\ell \to r$ to an \egraph,
  \texttt{ematch}
  finds tuples $(\sigma, c)$ where \eclass $c$ represents $\ell[\sigma]$.
Then, for each tuple,
  \mbox{\texttt{merge($c$, add($r[\sigma]$))}} adds $r[\sigma]$ to the \egraph
  and unifies it with the matching \eclass c.

\autoref{fig:egraph-rewrite} shows an \egraph undergoing a series of rewrites.
Note how the process is only additive; the initial term $(a \times 2) / 2$ is
  still represented in the \egraph.
Rewriting in an \egraph can also saturate, meaning the \egraph has
  learned every possible equivalence derivable from the given rewrites.
If the user tried to apply $x \times y \to y \times x$ to an \egraph twice,
  the second time would add no additional \enodes and perform no new merges;
  the \egraph can detect this and stop applying that rule.

\subsection{Equality Saturation}
\label{sec:eqsat}

Term rewriting~\cite{nachum-rewrites} is a time-tested approach
  for equational reasoning in
  program optimization~\cite{eqsat, denali},
  theorem proving~\cite{simplify, z3},
  and program transformation~\cite{graphs}.
In this setting, a tool repeatedly chooses one of a set of axiomatic rewrites,
  searches for matches of the left-hand pattern in the given
  expression, and replaces matching instances with the substituted
  right-hand side.

Term rewriting is typically destructive and ``forgets'' the matched
  left-hand side.
Consider applying a simple strength reduction rewrite:
  ${ (a \times 2) / 2 \to (a \ll 1) / 2 }$.
The new term carries no
  information about the initial term.
Applying strength reduction at this point prevents us from canceling out $2/2$.
In the compilers community, this classically tricky question of when to apply
  which rewrite is called the \textit{phase ordering} problem.

One solution to the phase ordering problem would simply apply all
  rewrites simultaneously, keeping track of every expression seen.
This eliminates the problem of choosing the right rule, but
  a naive implementation would require space exponential in the number
  of given rewrites.
\textit{Equality saturation}~\cite{eqsat, eqsat-llvm} is a technique to do this
  rewriting efficiently using an \egraph.

\begin{figure}
  \begin{minipage}{0.48\linewidth}
    \centering
    \includegraphics[width=0.9\linewidth]{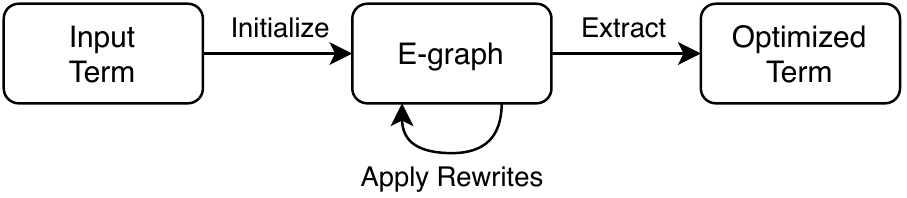}
  \end{minipage}
  \hfill
  \begin{minipage}{0.46\linewidth}
  \begin{lstlisting}[language=Python, gobble=4, numbers=left, basicstyle=\scriptsize\ttfamily]
    def equality_saturation(expr, rewrites):
      egraph = initial_egraph(expr)

      while not egraph.is_saturated_or_timeout():

        for rw in rewrites:
          for (subst, eclass) in egraph.ematch(rw.lhs):
            eclass2 = egraph.add(rw.rhs.subst(subst))
            egraph.merge(eclass, eclass2)

      return egraph.extract_best()
  \end{lstlisting}
  \end{minipage}
  \caption{
    Box diagram and pseudocode for equality saturation.
    Traditionally, equality saturation maintains the \egraph data structure
      invariants throughout the algorithm.
  }
  \label{fig:eq-sat-bg}
\end{figure}

\autoref{fig:eq-sat-bg} shows the equality saturation workflow.
First, an initial \egraph is created from the input term.
The core of the algorithm runs a set of rewrite rules until the \egraph is
  saturated (or a timeout is reached).
Finally, a procedure called \textit{extraction} selects the optimal represented
  term according to some cost function.
For simple cost functions, a bottom-up, greedy traversal of the \egraph suffices
  to find the best term.
Other extraction procedures have been explored for more complex cost
  functions~\cite{spores, wu_siga19}.

Equality saturation eliminates the tedious and often error-prone
  task of choosing when to apply which rewrites,
  promising an appealingly simple workflow: state the
  relevant rewrites for the language, create an initial \egraph from a given
  expression, fire the rules until saturation,
  and finally extract the cheapest equivalent expression.
Unfortunately, the technique remains ad hoc; prospective equality saturation
  users must implement their own \egraphs customized to their language, avoid
  performance pitfalls, and hack in the ability to do interpreted reasoning
  that is not supported by purely syntactic rewrites.
\egg aims to address each aspect of these difficulties.

\subsection{Equality Saturation and Theorem Proving}

An equality saturation engine and a theorem prover each have capabilities that
  would be impractical to replicate in the other.
Automated theorem provers like satisfiability modulo theory (SMT) solvers are
  general tools that, in addition to supporting satisfiability queries,
  incorporate sophisticated, domain-specific solvers to allow interpreted
  reasoning within the supported theories.
On the other hand, equality saturation is specialized for optimization, and its
  extraction procedure directly produces an optimal term with respect to a given
  cost function.


While SMT solvers are indeed the more general tool,
  equality saturation is not superseded by SMT;
  the specialized approach can be much faster when the full generality of SMT is
  not needed.
To demonstrate this, we replicated a portion of the recent TASO paper~\cite{taso},
  which optimizes deep learning models.
As part of the work, they must verify a set of synthesized equalities with
  respect to a trusted set of universally quantified axioms.
TASO uses Z3~\cite{z3} to perform the
  verification even though most of Z3's features
  (disjunctions, backtracking, theories, etc.)
  were not required.
An equality saturation engine can also be used for verifying these equalities
  by adding the left and right sides of
  each equality to an \egraph,
  running the axioms as rewrites,
  and then checking if both sides end up in the same \eclass.
Z3 takes 24.65 seconds to perform the verification;
  \egg performs the same task in 1.56 seconds ($15\times$ faster),
  or only 0.52 seconds ($47\times$ faster) when using
  \egg's batched evaluation (\autoref{sec:egg-batched}).




\section{Rebuilding: A New Take on \Egraph Invariant Maintenance}
\label{sec:rebuild}
\label{sec:rebuilding}

Traditionally~\cite{nelson, simplify},
  \egraphs maintain their data structure invariants
  after each operation.
We separate this invariant restoration into a procedure called \textit{rebuilding}.
This separation allows the client to
  choose when to enforce the \egraph invariants.
Performing a rebuild immediately after every operation replicates the
  traditional approach to invariant maintenance.
In contrast, rebuilding less frequently can amortize the cost of invariant
  maintenance, significantly improving performance.

In this section, we first describe how e-graphs have
  traditionally maintained invariants (\autoref{sec:upward}).
We then describe the rebuilding framework and how it captures a spectrum of
  invariant maintenance approaches, including the traditional one
  (\autoref{sec:rebuilding-detail}).
Using this flexibility, we then give a modified algorithm for equality
  saturation that enforces the \egraph invariants at only select points
  (\autoref{sec:rebuilding-eqsat}).
We finally demonstrate that this new approach offers an asymptotic speedup over
  traditional equality saturation (\autoref{sec:rebuild-eval}).






\subsection{Upward Merging}
\label{sec:upward}

Both mutating operations on the \egraph
  (\texttt{add} and \texttt{merge}, \autoref{sec:interface})
  can break the \egraph invariants if not done carefully.
\Egraphs have traditionally used \textit{hashconsing} and
  \textit{upward merging} to maintain the congruence invariant.

The \texttt{add} operation relies on the hashcons invariant
  (Definition \ref{def:hash-inv})
  to quickly check whether the \enode $n$ to be added---or one congruent to it---is
  already present.
Without this check, \texttt{add} would create a new \eclass with $n$ in it
  even if some $n' \cong n$ was already in the \egraph,
  violating the congruence invariant.

The \texttt{merge} operation \eclasses can violate both \egraph invariants.
If $f(a, b)$ and $f(a, c)$ reside in two different \eclasses $x$ and $y$,
  merging $b$ and $c$ should also merge $x$ and $y$ to maintain the congruence invariant.
This can propagate further, requiring additional merges.


\Egraphs maintain a \textit{parent list} for each \eclass
  to maintain congruence.
The parent list for \eclass $c$ holds all \enodes that have $c$ as a child.
When merging two \eclasses, \egraphs inspect these parent lists to find parents
  that are now congruent, recursively ``upward merging'' them if necessary.

The \texttt{merge} routine must also perform bookkeeping to preserve the
  hashcons invariant.
In particular, merging two \eclasses may change how parent \enodes of those
  \eclasses are canonicalized.
The \texttt{merge} operation must therefore
  remove, re-canonicalize, and replace those \enodes in the
  hashcons.
In existing \egraph implementations~\cite{herbie} used for equality saturation,
  maintaining the invariants while merging can take the vast majority of
  run time.

\subsection{Rebuilding in Detail}
\label{sec:rebuilding-detail}

\begin{figure}
  \begin{minipage}[t]{0.47\linewidth}
    \begin{lstlisting}[gobble=4, numbers=left, basicstyle=\scriptsize\ttfamily, escapechar=|]
    def add(enode):
      enode = self.canonicalize(enode)
      if enode in self.hashcons:
        return self.hashcons[enode]
      else:
        eclass_id = self.new_singleton_eclass(enode)
        for child in enode.children:
          child.parents.add(enode, eclass_id)
        self.hashcons[enode] = eclass_id
        return eclass_id

    def merge(id1, id2)
      if self.find(id1) == self.find(id2):
        return self.find(id1)
      new_id = self.union_find.union(id1, id2)
      # traditional egraph merge can be
      # emulated by calling rebuild right after
      # adding the eclass to the worklist
      self.worklist.add(new_id) |\label{line:worklist-add}|
      return new_id

    def canonicalize(enode) |\label{line:canon}|
      new_ch = [self.find(e) for e in enode.children]
      return mk_enode(enode.op, new_ch)

    def find(eclass_id):
      return self.union_find.find(eclass_id)
    \end{lstlisting}
  \end{minipage}
  \hfill
  \begin{minipage}[t]{0.47\linewidth}
    \begin{lstlisting}[gobble=4, numbers=left, firstnumber=27, basicstyle=\scriptsize\ttfamily]
    def rebuild():
      while self.worklist.len() > 0:
        # empty the worklist into a local variable
        todo = take(self.worklist)
        # canonicalize and deduplicate the eclass refs
        # to save calls to repair
        todo = { self.find(eclass) for eclass in todo }
        for eclass in todo:
          self.repair(eclass)

    def repair(eclass):
      # update the hashcons so it always points
      # canonical enodes to canonical eclasses
      for (p_node, p_eclass) in eclass.parents:
        self.hashcons.remove(p_node)
        p_node = self.canonicalize(p_node)
        self.hashcons[p_node] = self.find(p_eclass)

      # deduplicate the parents, noting that equal
      # parents get merged and put on the worklist
      new_parents = {}
      for (p_node, p_eclass) in eclass.parents:
        p_node = self.canonicalize(p_node)
        if p_node in new_parents:
          self.merge(p_eclass, new_parents[p_node])
        new_parents[p_node] = self.find(p_eclass)
      eclass.parents = new_parents
    \end{lstlisting}
  \end{minipage}
  \caption{
    Pseudocode for the \texttt{add}, \texttt{merge}, \texttt{rebuild}, and
    supporting methods.
    In each method, \texttt{self} refers to the \egraph being modified.
  }
  \label{fig:rebuild-code}
\end{figure}

Traditionally, invariant restoration is part of the
  \texttt{merge} operation itself.
Rebuilding separates these concerns,
  reducing \texttt{merge}'s obligations
  and allowing for amortized invariant maintenance.
In the rebuilding paradigm,
  \texttt{merge} maintains a \textit{worklist} of \eclass ids that need to
  be ``upward merged'', i.e., \eclasses whose parents are possibly congruent but
  not yet in the same \eclass.
The \texttt{rebuild} operation processes this worklist, restoring the invariants
  of deduplication and congruence.
Rebuilding is similar to other approaches in how it restores congruence
  (see \nameref{sec:related} for comparison to \citet{downey-cse});
  but it uniquely allows the client to choose when to restore invariants in the
  context of a larger algorithm like equality saturation.

\autoref{fig:rebuild-code} shows pseudocode for the main \egraph operations and
  rebuilding.
Note that \texttt{add} and \texttt{canonicalize} are given for completeness, but
  they are unchanged from the traditional \egraph implementation.
The \texttt{merge} operation is similar, but it only adds the new \eclass to the
  worklist instead of immediately starting upward merging.
Adding a call to \texttt{rebuild} right after the addition to
  the worklist (\autoref{fig:rebuild-code} line \ref{line:worklist-add})
  would yield the traditional behavior of restoring the invariants immediately.

The \texttt{rebuild} method essentially calls \texttt{repair} on the \eclasses
  from the worklist until the worklist is empty.
Instead of directly manipulating the worklist, \egg's \texttt{rebuild} method
  first moves it into a local variable and deduplicates \eclasses
  up to equivalence.
Processing the worklist may \texttt{merge} \eclasses,
  so breaking the worklist into chunks ensures that \eclass ids made
  equivalent in the previous chunk are deduplicated in the subsequent chunk.

The actual work of \texttt{rebuild} occurs in the \texttt{repair} method.
\texttt{repair} examines an \eclass $c$ and first canonicalizes \enodes in the
  hashcons that have $c$ as a child.
Then it performs what is essentially one ``layer'' of upward
  merging:
if any of the parent \enodes have become congruent, then their
  \eclasses are merged and the result is added to the worklist.

Deduplicating the worklist, and thus reducing calls to \texttt{repair},
  is at the heart of why deferring rebuilding improves
  performance.
Intuitively, the upward merging process of rebuilding traces out a ``path'' of
  congruence through the \egraph.
When rebuilding happens immediately after \texttt{merge}
  (and therefore frequently), these paths can substantially overlap.
By deferring rebuilding, the chunk-and-deduplicate approach can coalesce the
overlapping parts of these paths, saving what would have been redundant work.
In our modified equality saturation algorithm (\autoref{sec:rebuilding-eqsat}),
  deferred rebuilding is responsible for a significant, asymptotic speedup
  (\autoref{sec:rebuild-eval}).

\subsubsection{Examples of Rebuilding}

Deferred rebuilding speeds up congruence maintenance by amortizing the work of
  maintaining the hashcons invariant.
Consider the following terms in an \egraph:
  $f_{1}(x), ..., f_{n}(x),\, y_{1}, ..., y_{n}$.
Let the workload be $\texttt{merge}(x, y_{1}), ..., \texttt{merge}(x, y_{n})$.
Each merge may change the canonical representation of the $f_{i}(x)$s,
  so the traditional invariant maintenance strategy
  could require $O(n^{2})$ hashcons updates.
With deferred rebuilding the \texttt{merges} happen before
  the hashcons invariant is restored,
  requiring no more than $O(n)$ hashcons updates.

Deferred rebuilding can also reduce the number of calls to \texttt{repair}.
Consider the following $w$ terms in an \egraph,
  each nested under $d$ function symbols:
  $$f_1 (f_2(\ldots f_d(x_1))), \quad\ldots,\quad f_1(f_2(\ldots f_d(x_w)))$$
Note that $w$ corresponds the width of this group of terms, and $d$ to the depth.
Let the workload be $w-1$ merges that merge all the $x$s together:
  for $i \in [2, w], \texttt{merge}(x_{1}, x_{i})$.

In the traditional upward merging paradigm
  where \texttt{rebuild} is called after every \texttt{merge},
  each $\texttt{merge}(x_i, x_j)$ will require $O(d)$ calls to \texttt{repair}
  to maintain congruence, one for each layer of $f_{i}$s.
Over the whole workload, this requires $O(wd)$ calls to \texttt{repair}.

With deferred rebuilding, however, the $w-1$ merges can all take place before
  congruence must be restored.
Suppose the $x$s are all merged into an \eclass $c_{x}$
When \texttt{rebuild} finally is called,
  the only element in the deduplicated worklist is $c_{x}$.
Calling \texttt{repair} on $c_{x}$ will merge the \eclasses of the $f_{d}$
  \enodes into an \eclass $c_{f_{d}}$,
  adding the \eclasses that contained those \enodes back to the worklist.
When the worklist is again deduplicated,
  $c_{f_{d}}$ will be the only element,
  and the process repeats.
Thus, the whole workload only incurs $O(d)$ calls to \texttt{repair},
  eliminating the factor corresponding to the width of this group of terms.
\autoref{fig:repair-plot} shows that the number calls to \texttt{repair} is
  correlated with time spent doing congruence maintenance.

\subsubsection{Proof of Congruence}

Intuitively, rebuilding is a delay of the upward merging process, allowing
  the user to choose when to restore the \egraph invariants.
They are substantially similar in structure, with a critical a difference in when
  the code is run.
Below we offer a proof demonstrating that rebuilding restores the
\egraph congruence invariant.

\begin{theorem}
  Rebuilding restores congruence and terminates.
\end{theorem}

\begin{proof}
  Since rebuilding only merges congruent nodes,
    the congruence closure $\cong^{*}$ is fixed even though $\equivnode$ changes.
  When $(\equivnode) = (\cong^*)$, congruence is restored.
  Note that both $\equivnode$ and $\cong^*$ are finite.
  We therefore show that rebuilding causes $\equivnode$ to approach $\cong^*$.
  We define the set of incongruent \enode pairs as $I = (\cong^*) \setminus (\equivnode)$;
  in other words,
    $(n_{1}, n_{2}) \in I$ if $n_{1} \cong^{*} n_{2}$
     but $n_{1} \not\equivnode n_{2}$.

  Due to the additive nature of equality saturation, $\equivnode$ only increases
    and therefore $I$ is non-increasing.
  However, a call to \texttt{repair} inside the loop of \texttt{rebuild} does
    not necessarily shrink $I$.
  Some calls instead remove an element from the worklist but do not modify the
    \egraph at all.

  Let the set $W$ be the worklist of \eclasses to be processed by
    \texttt{repair};
  in \autoref{fig:rebuild-code}, $W$ corresponds to \texttt{self.worklist} plus
    the unprocessed portion of the \texttt{todo} local variable.
  We show that each call to \texttt{repair} decreases the tuple
    $(|I|, |W|)$ lexicographically until $(|I|, |W|) = (0, 0)$,
    and thus rebuilding terminates with $(\equivnode) = (\cong^*)$.


  Given an \eclass $c$ from $W$, \texttt{repair} examines $c$'s parents
    for congruent \enodes that are not yet in the same \eclass:
  \begin{itemize}
    \item If at least one pair of $c$'s parents are congruent,
          rebuilding merges each pair $(p_{1}$, $p_{2})$,
          which adds to $W$ but makes $I$ smaller by definition.
    \item If no such congruent pairs are found, do nothing.
          Then, $|W|$ is decreased by 1 since $c$ came from the
          worklist and \texttt{repair} did not add anything back.
  \end{itemize}

  Since $(|I|, |W|)$ decreases lexicographically,
    $|W|$ eventually reaches $0$, so \texttt{rebuild} terminates.
  Note that $W$ contains precisely those \eclasses that need to be
    ``upward merged'' to check for congruent parents.
  So, when $W$ is empty,
    \texttt{rebuild} has effectively performed upward merging.
  By~\citet[Chapter 7]{nelson}, $|I| = 0$.
Therefore, when rebuilding terminates, congruence is restored.

\end{proof}

\subsection{Rebuilding and Equality Saturation}
\label{sec:rebuilding-eqsat}

Rebuilding offers the choice of when to enforce the \egraph invariants,
  potentially saving work if deferred thanks to the deduplication of the
  worklist.
The client is responsible for rebuilding at a time that
  maximizes performance without limiting the application.

\begin{figure}
  \begin{subfigure}[t]{0.47\linewidth}
    \begin{lstlisting}[language=Python, gobble=6, numbers=left, basicstyle=\scriptsize\ttfamily]
      def equality_saturation(expr, rewrites):
        egraph = initial_egraph(expr)

        while not egraph.is_saturated_or_timeout():


          # reading and writing is mixed
          for rw in rewrites:
            for (subst, eclass) in egraph.ematch(rw.lhs):

              # in traditional equality saturation,
              # matches can be applied right away
              # because invariants are always maintained
              eclass2 = egraph.add(rw.rhs.subst(subst))
              egraph.merge(eclass, eclass2)

              # restore the invariants after each merge
              egraph.rebuild()

        return egraph.extract_best()
    \end{lstlisting}
    \caption{
      Traditional equality saturation alternates between searching and applying
      rules, and the \egraph maintains its invariants throughout.
    }
    \label{fig:eq-sat-code1}
  \end{subfigure}
  \hfill
  \begin{subfigure}[t]{0.47\linewidth}
    \begin{lstlisting}[language=Python, gobble=6, basicstyle=\scriptsize\ttfamily, numbers=left]
      def equality_saturation(expr, rewrites):
        egraph = initial_egraph(expr)

        while not egraph.is_saturated_or_timeout():
          matches = []

          # read-only phase, invariants are preserved
          for rw in rewrites:
            for (subst, eclass) in egraph.ematch(rw.lhs):
              matches.append((rw, subst, eclass))

          # write-only phase, temporarily break invariants
          for (rw, subst, eclass) in matches:
            eclass2 = egraph.add(rw.rhs.subst(subst))
            egraph.merge(eclass, eclass2)

          # restore the invariants once per iteration
          egraph.rebuild()

        return egraph.extract_best()
    \end{lstlisting}
    \caption{
      \egg splits equality saturation iterations into read and write phases.
      The \egraph invariants are not constantly maintained, but restored
      only at the end of each iteration by the \texttt{rebuild} method
      (\autoref{sec:rebuild}).
    }
    \label{fig:eq-sat-code2}
  \end{subfigure}

  \caption{
    Pseudocode for traditional and \egg's version of the equality saturation
    algorithm.
  }
  \label{fig:eq-sat-code}
\end{figure}

\egg provides a modified equality saturation algorithm to take advantage
  of rebuilding.
\autoref{fig:eq-sat-code} shows pseudocode for both traditional equality
  saturation and \egg's variant, which exhibits two key differences:
\begin{enumerate}
  \item Each iteration is split into a read phase, which searches for all the
        rewrite matches, and a write phase that applies those matches.\footnote
    {
      Although the original equality saturation paper~\cite{eqsat}
      does not have separate reading and writing phases,
      some \egraph implementations (like the one inside Z3~\cite{z3})
      do separate these phases as an implementation detail.
      Ours is the first algorithm to take advantage of this by deferring
      invariant maintenance.
    }
  \item Rebuilding occurs only once per iteration, at the end.
\end{enumerate}

\egg's separation of the read and write phases means that rewrites are truly
  unordered.
In traditional equality saturation, later rewrites in the given rewrite list are
  favored in the sense that they can ``see'' the results of earlier rewrites in
  the same iteration.
Therefore, the results depend on the order of the rewrite list
  if saturation is not reached (which is common on large rewrite lists or input
  expressions).
\egg's equality saturation algorithm is invariant to the order of the rewrite
  list.

Separating the read and write phases also allows \egg to safely defer rebuilding.
If rebuilding were deferred in the traditional equality saturation algorithm,
  rules later in the rewrite list would be searched against an \egraph with
  broken invariants.
Since congruence may not hold, there may be missing equivalences, resulting in
  missing matches.
These matches will be seen after the \texttt{rebuild} during the next iteration
  (if another iteration occurs), but the false reporting could impact metrics
  collection, rule scheduling,\footnotemark{} or saturation detection.
\footnotetext{
  An optimization introduced in \autoref{sec:rule-scheduling} that
  relies on an accurate count of how many times a rewrite was matched.
}

\subsection{Evaluating Rebuilding}
\label{sec:rebuild-eval}

To demonstrate that deferred rebuilding
  provides faster congruence closure than traditional upward merging,
  we modified \egg to call \texttt{rebuild} immediately after every \texttt{merge}.
This provides a one-to-one comparison of deferred rebuilding against the
  traditional approach, isolated
  from the many other factors that make \egg efficient: overall design
  and algorithmic differences, programming language performance, and other
  orthogonal performance improvements.

\begin{figure}
  \begin{subfigure}{0.49\linewidth}
    \includegraphics[height=5cm]{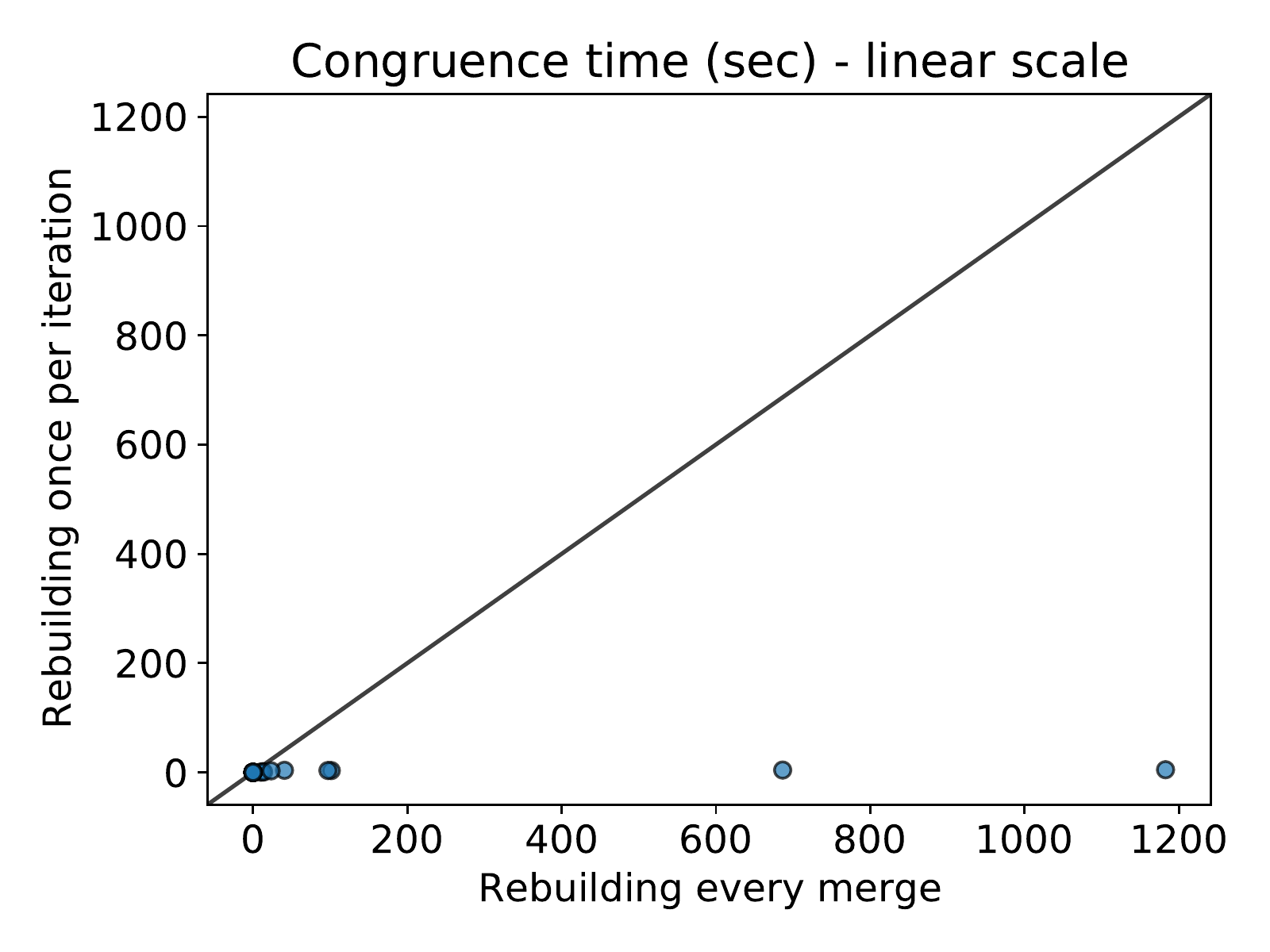}
  \end{subfigure}
  \hfill
  \begin{subfigure}{0.49\linewidth}
    \includegraphics[height=5cm]{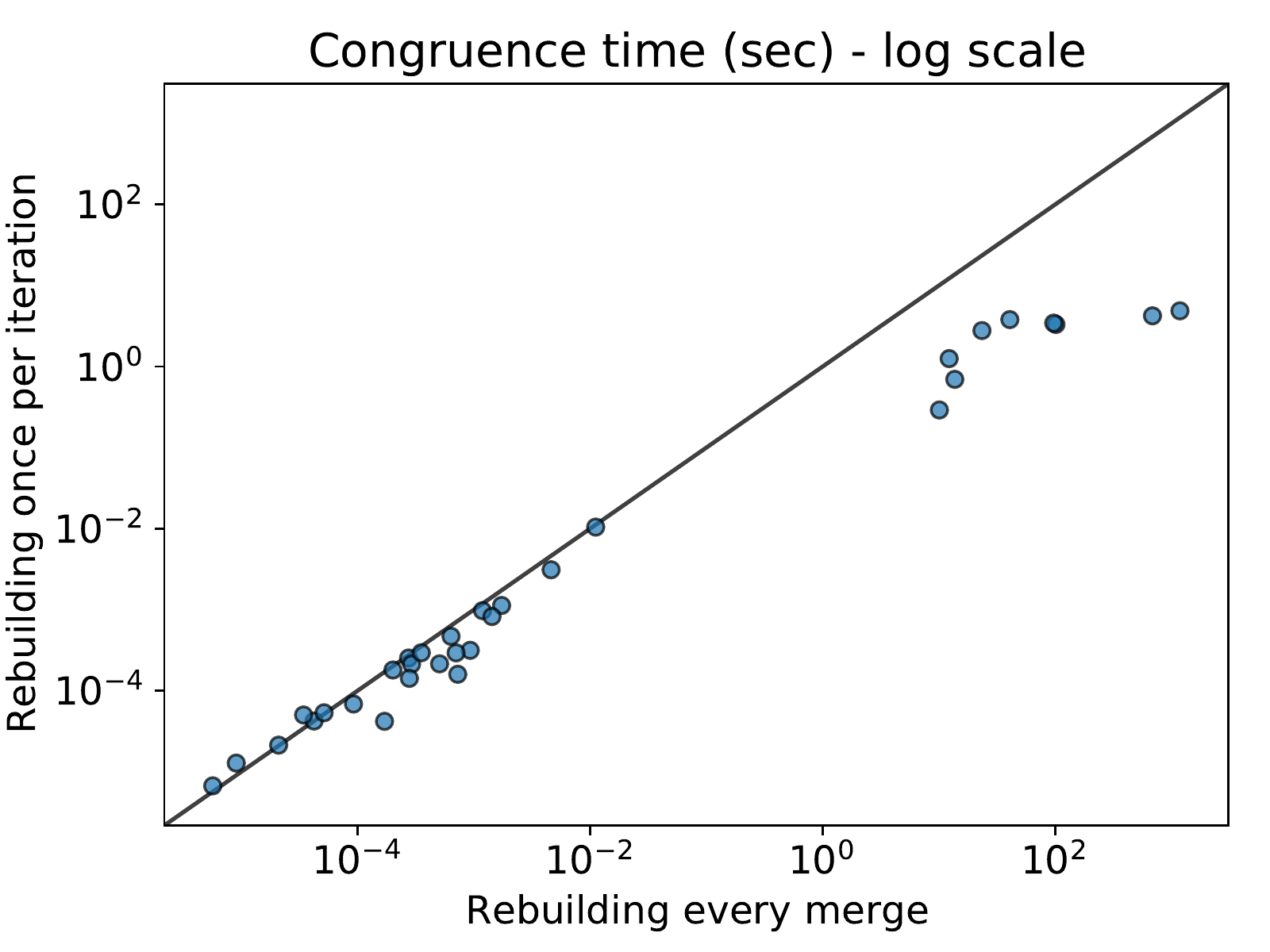}
  \end{subfigure}
  \caption{
    Rebuilding once per iteration---as opposed to after every merge---significantly
      speeds up congruence maintenance.
    Both plots show the same data: one point for each of the \nEggTests tests.
    The diagonal line is $y=x$;
      points below the line mean deferring rebuilding is faster.
    In aggregate over all tests (using geometric mean),
      congruence is \CongrSpeedup faster, and
      equality saturation is \TotalSpeedup faster.
    The linear scale plot shows that deferred rebuilding is significantly faster.
    The log scale plot suggests the speedup is greater than some constant multiple;
      \autoref{fig:eval-iter} demonstrates this in greater detail.
  }
  \label{fig:eval}

  \begin{minipage}[t]{0.48\linewidth}
  \includegraphics[height=5cm]{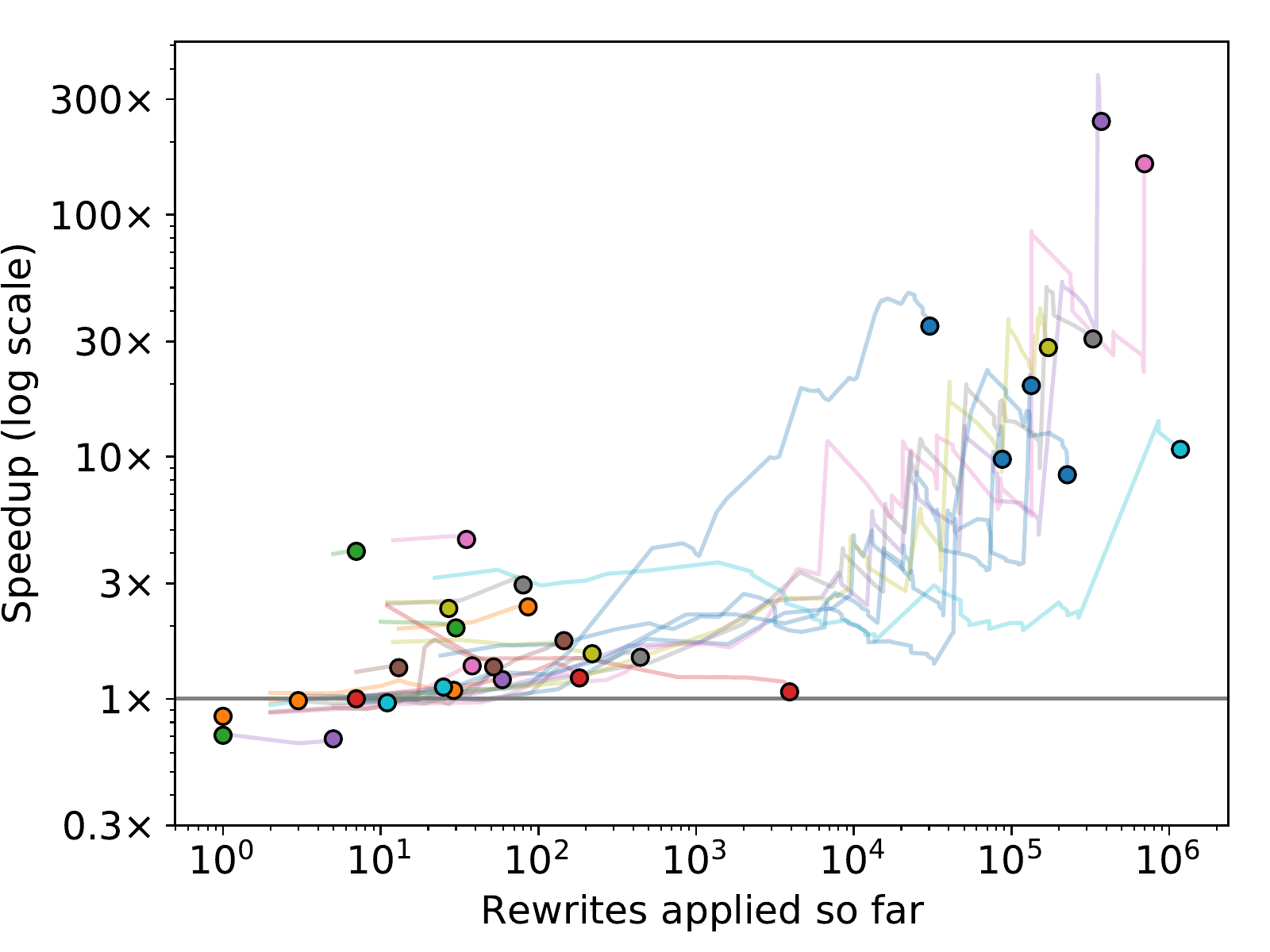}
  \caption{
    As more rewrites are applied, deferring rebuilding gives greater speedup.
    Each line represents a single test: each equality saturation iteration plots
      the cumulative rewrites applied so far against the multiplicative speedup
      of deferring rebuilding; the dot represents the end of that test.
    Both the test suite as a whole (the dots) and individual tests (the lines)
      demonstrate an asymptotic speedup that increases with
      the problem size.
  }
  \label{fig:eval-iter}
  \end{minipage}
  \hfill
  \begin{minipage}[t]{0.48\linewidth}
  \includegraphics[height=5cm]{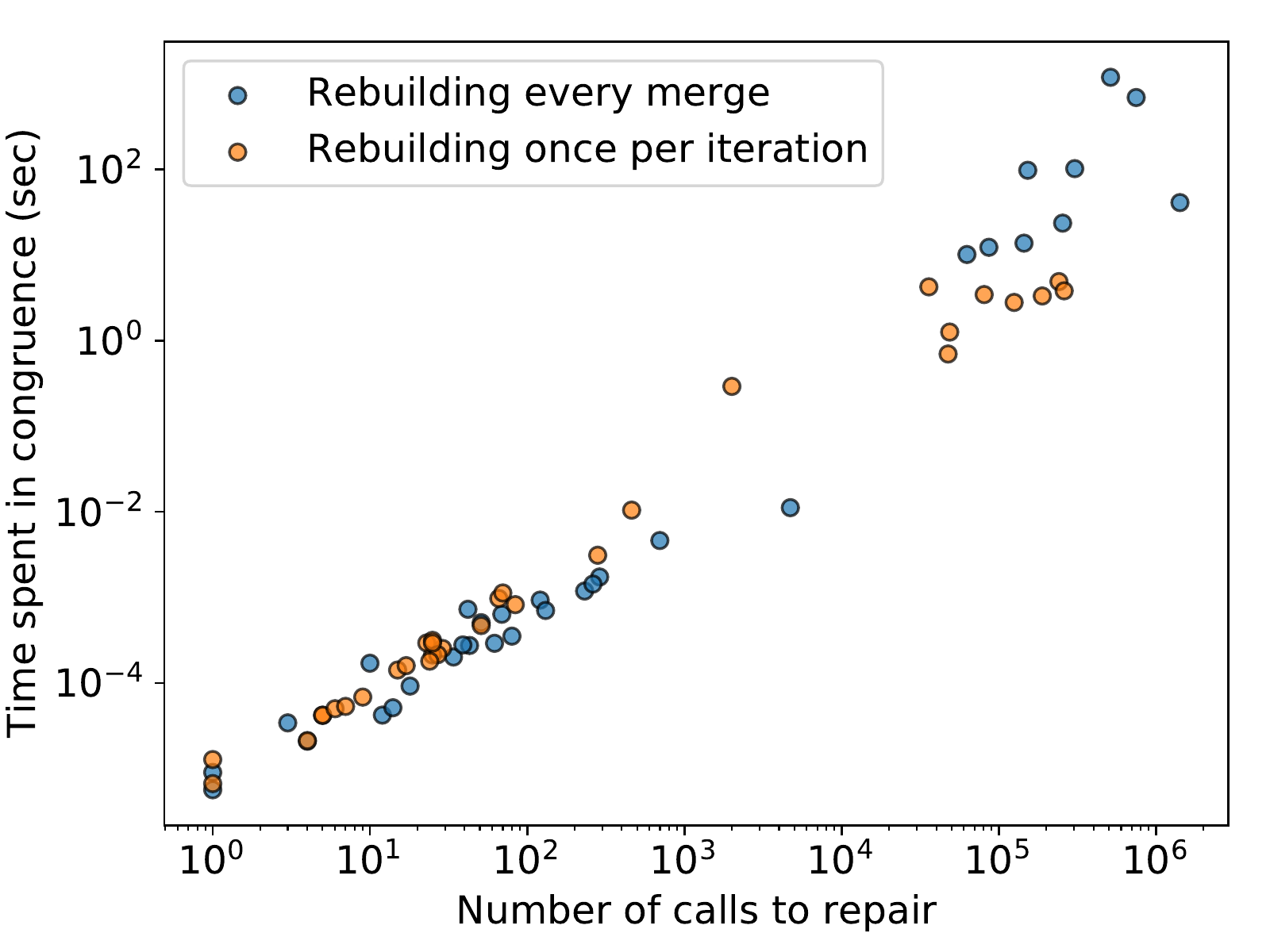}
  \caption{
    The time spent in congruence maintenance correlates with the number of calls
    to the \texttt{repair} method.
    Spearman correlation yields $r=\RepairsR$ with a p-value of \RepairsP,
    indicating that the two quantities are indeed positively correlated.
  }
  \label{fig:repair-plot}
  \end{minipage}
\end{figure}

We ran \egg's test suite using both rebuild strategies, measuring the time spent
  on congruence maintenance.
Each test consists of one run of \egg's equality saturation algorithm to optimize
  a given expression.
Of the \nEggTests total tests,
  \nEggTimeouts hit the iteration limit of 100 and the remainder saturated.
Note that both rebuilding strategies use \egg's phase-split equality saturation
  algorithm, and the resulting \egraphs are identical in all cases.
These experiments were performed on a 2020 Macbook Pro with a 2 GHz quad-core
  Intel Core i5 processor and 16GB of memory.

\autoref{fig:eval} shows our how rebuilding speeds up congruence maintenance.
Overall, our experiments show an aggregate \CongrSpeedup speedup on congruence
  closure and \TotalSpeedup speedup over the entire equality saturation
  algorithm.
\autoref{fig:eval-iter} shows this speedup is asymptotic;
  the multiplicative speedup increases as problem gets larger.

\egg's test suite consists of two main applications:
\texttt{math},
  a small computer algebra system capable of symbolic differentiation and
  integration; and
\texttt{lambda},
  a partial evaluator for the untyped lambda calculus using explicit
  substitution to handle variable binding (shown in \autoref{sec:impl}).
Both are typical \egg applications primarily driven by
  syntactic rewrites, with a few key uses of \egg's more complex features
  like \eclass analyses and dynamic/conditional rewrites.

\egg can be configured to capture various metrics about equality saturation as
  it runs, including the time spent in the read phase (searching for matches),
  the write phase (applying matches), and rebuilding.
In \autoref{fig:eval}, congruence time is measured as the time spent
  applying matches plus rebuilding.
Other parts of the equality saturation algorithm (creating the initial \egraph,
  extracting the final term) take negligible take compared to the equality
  saturation iterations.

Deferred rebuilding amortizes the examination of \eclasses
  for congruence maintenance;
  deduplicating the worklist reduces the number of calls to the \texttt{repair}.
\autoref{fig:repair-plot} shows that time spent in congruence is correlated with
  the number of calls to the \texttt{repair} methods.

The case study in \autoref{sec:herbie} provides a further evaluation of
  rebuilding. Rebuilding (and other \egg features) have also been implemented in
  a Racket-based \egraph, demonstrating that rebuilding is a conceptual advance
  that need not be tied to the \egg implementation.


\section{Extending \Egraphs with \Eclass Analyses}
\label{sec:extensions}

As discussed so far, \egraphs and equality saturation provide an efficient way
  to implement a term rewriting system.
Rebuilding enhances that efficiency, but the approach remains designed for
  purely syntactic rewrites.
However, program analysis and optimization typically require more than just
  syntactic information.
Instead, transformations are \emph{computed} based on the input terms and also semantic facts
  about that input term, e.g., constant value, free variables, nullability,
  numerical sign, size in memory, and so on.
The ``purely syntactic'' restriction has forced existing equality saturation
  applications~\cite{eqsat, eqsat-llvm, herbie} to
  resort to ad hoc passes over the \egraph
  to implement analyses like constant folding.
These ad hoc passes require manually manipulating the \egraph,
  the complexity of which could prevent the implementation of more sophisticated
  analyses.

We present a new technique called \textit{\eclass analysis},
  which allows the concise
  expression of a program analysis over the \egraph.
An \eclass analysis resembles abstract interpretation
  lifted to the \egraph level,
  attaching \textit{analysis data} from a semilattice to each \eclass.
The \egraph maintains and propagates this data as
  \eclasses get merged and new \enodes are added.
Analysis data can be used directly to modify the \egraph, to inform
  how or if rewrites apply their right-hand sides, or to determine the cost of
  terms during the extraction process.

\Eclass analyses provide a general mechanism to replace what previously
  required ad hoc extensions that manually manipulate the \egraph.
\Eclass analyses also fit within the equality saturation workflow,
  so they can naturally cooperate with the equational reasoning provided by
  rewrites.
Moreover, an analysis lifted to the \egraph level automatically benefits from a
  sort of ``partial-order reduction'' for free:
  large numbers of similar programs may be analyzed for little additional cost
  thanks to the \egraph's compact representation.

This section provides a conceptual explanation of \eclass analyses as well
  as dynamic and conditional rewrites that can use the analysis data.
The following sections will provide concrete examples:
  \autoref{sec:impl} discusses the \egg implementation and a complete example of a
  partial evaluator for the lambda calculus;
  \autoref{sec:case-studies} discusses how three published projects have used
  \egg and its unique features (like \eclass analyses).





\subsection{\Eclass Analyses}
\label{sec:analysis}

An \eclass analysis defines a domain $D$ and associates a value $d_{c} \in D$ to
  each \eclass $c$.
The \eclass $c$ contains the associated data $d_{c}$,
  i.e., given an \eclass $c$, one can get $d_{c}$ easily, but not vice-versa.

The interface of an \eclass analysis is as follows,
  where $G$ refers to the \egraph,
  and $n$ and $c$ refer to \enodes and \eclasses within $G$:


\vspace{1em}
\begin{tabular}{lp{0.7\linewidth}}
  $\textsf{make}(n) \to d_{c}$ &
    When a new \enode $n$ is added to $G$ into a new, singleton \eclass $c$,
    construct a new value $d_{c} \in D$ to be associated with $n$'s new \eclass,
    typically by accessing the associated data of $n$'s children.
  \\
  $\textsf{join}(d_{c_1}, d_{c_2}) \to d_{c}$ &
    When \eclasses $c_{1}, c_{2}$ are being merged into $c$,
    join $d_{c_1}, d_{c_2}$ into a new value $d_{c}$ to be associated with the
    new \eclass $c$.
  \\
  $\textsf{modify}(c) \to c'$ &
    Optionally modify the \eclass $c$ based on $d_{c}$, typically by adding an
      \enode to $c$.
    Modify should be idempotent if no other changes occur to the \eclass, i.e.,
      $\textsf{modify}(\textsf{modify}(c)) = \textsf{modify}(c)$
\end{tabular}
\vspace{1em}

The domain $D$ together with the \textsf{join} operation should form a join-semilattice.
The semilattice perspective is useful for defining the \textit{analysis invariant}
  (where $\wedge$ is the \textsf{join} operation):
\[
  \forall c \in G.\quad
  d_{c} = \bigwedge_{n \in c} \textsf{make}(n)
  \quad \text{and} \quad
  \textsf{modify}(c) = c
\]

The first part of the analysis invariant states that the data associated with
  each \eclass must be the \textsf{join} of the \textsf{make} for every \enode
  in that \eclass.
Since $D$ is a join-semilattice, this means that
  $\forall c, \forall n \in c, d_{c} \geq \textsf{make}(n) $.
The motivation for the second part is more subtle.
Since the analysis can modify an \eclass through the \textsf{modify} method,
  the analysis invariant asserts that these modifications are driven to a fixed
  point.
When the analysis invariant holds, a client looking at the analysis data can be
  assured that the analysis is ``stable'' in the sense that
  recomputing \textsf{make}, \textsf{join}, and \textsf{modify} will not
  modify the \egraph or any analysis data.

\subsubsection{Maintaining the Analysis Invariant}

\begin{figure}
  \begin{minipage}[t]{0.47\linewidth}
    \begin{lstlisting}[gobble=4, numbers=left, numberstyle=\color{black}, basicstyle=\scriptsize\ttfamily\color{black!40}, escapechar=|]
    def add(enode):
      enode = self.canonicalize(enode)
      if enode in self.hashcons:
        return self.hashcons[enode]
      else:
        eclass = self.new_singleton_eclass(enode)
        for child_eclass in enode.children:
          child_eclass.parents.add(enode, eclass)
        self.hashcons[enode] = eclass
        |\color{black}\label{line:add1} eclass.data = analysis.make(enode)|
        |\color{black}\label{line:add2} analysis.modify(eclass)|
        return eclass

    def merge(eclass1, eclass2)
      union = self.union_find.union(eclass1, eclass2)
      if not union.was_already_unioned:
        |\color{black}\label{line:merge1}d1, d2 = eclass1.data, eclass2.data|
        |\color{black}\label{line:merge2}union.eclass.data = analysis.join(d1, d2)|
        self.worklist.add(union.eclass)
      return union.eclass
    \end{lstlisting}
  \end{minipage}
  \hfill
  \begin{minipage}[t]{0.47\linewidth}
    \begin{lstlisting}[gobble=4, numbers=left, firstnumber=21, numberstyle=\color{black}, basicstyle=\scriptsize\ttfamily\color{black!40}, escapechar=|]
    def repair(eclass):
      for (p_node, p_eclass) in eclass.parents:
        self.hashcons.remove(p_node)
        p_node = self.canonicalize(p_node)
        self.hashcons[p_node] = self.find(p_eclass)

      new_parents = {}
      for (p_node, p_eclass) in eclass.parents:
        p_node = self.canonicalize(p_node)
        if p_node in new_parents:
          self.union(p_eclass, new_parents[p_node])
        new_parents[p_node] = self.find(p_eclass)
      eclass.parents = new_parents
    \end{lstlisting}
    \vspace{-3mm}
    \begin{lstlisting}[gobble=4, numbers=left, firstnumber=34, basicstyle=\scriptsize\ttfamily, escapechar=|]

      # any mutations modify makes to eclass
      # will add to the worklist
      |\label{line:repair1}|analysis.modify(eclass)
      for (p_node, p_eclass) in eclass.parents:
        new_data = analysis.join(
          p_eclass.data,
          analysis.make(p_node))
        if new_data != p_eclass.data:
          p_eclass.data = new_data
          |\label{line:repair2}|self.worklist.add(p_eclass)
    \end{lstlisting}
  \end{minipage}
  \caption{
    The pseudocode for maintaining the \eclass analysis invariant is largely
      similar to how rebuilding maintains congruence closure
      (\autoref{sec:rebuilding}).
    Only lines \ref{line:add1}--\ref{line:add2},
      \ref{line:merge1}--\ref{line:merge2},
      and \ref{line:repair1}--\ref{line:repair2} are added.
    Grayed out or missing code is unchanged from \autoref{fig:rebuild-code}.
  }
  \label{fig:rebuild-analysis}
\end{figure}

We extend the rebuilding procedure from \autoref{sec:rebuilding} to restore the
  analysis invariant as well as the congruence invariant.
\autoref{fig:rebuild-analysis} shows the necessary modifications to the
  rebuilding code from \autoref{fig:rebuild-code}.

Adding \enodes and merging \eclasses risk breaking the analysis invariant in
  different ways.
Adding \enodes is the simpler case; lines \ref{line:add1}--\ref{line:add2}
  restore the invariant for the newly created, singleton \eclass that holds the
  new \enode.
When merging \enodes, the first concern is maintaining the semilattice portion of the
  analysis invariant.
Since \textsf{join} forms a semilattice over the domain $D$ of the analysis
  data, the order in which the joins occur does not matter.
Therefore, line \ref{line:merge2} suffices to update the analysis data of the
  merged \eclass.

Since $\textsf{make}(n)$ creates analysis data by looking at the data of $n$'s,
  children, merging \eclasses can violate the analysis invariant in the same way
  it can violate the congruence invariant.
The solution is to use the same worklist mechanism introduced in
  \autoref{sec:rebuilding}.
Lines \ref{line:repair1}--\ref{line:repair2} of the \texttt{repair} method
  (which \texttt{rebuild} on each element of the worklist)
  re-\textsf{make} and \textsf{merge} the analysis data of the parent of any
  recently merged \eclasses.
The new \texttt{repair} method also calls \textsf{modify} once, which suffices
  due to its idempotence.
In the pseudocode, \textsf{modify} is reframed as a mutating method for clarity.

\Egg's implementation of \eclass analyses assumes that the analysis domain $D$
  is indeed a semilattice and that \textsf{modify} is idempotent.
Without these properties, \egg may fail to restore the analysis invariant on
  \texttt{rebuild}, or it may not terminate.

\subsubsection{Example: Constant Folding}

The data produced by \eclass analyses can be
  usefully consumed by other components of an equality saturation system
  (see \autoref{sec:rewrites}),
  but \eclass analyses can be useful on their own thanks to the
  \textsf{modify} hook.
Typical \textsf{modify} hooks will either do nothing, check some invariant about
  the \eclasses being merged, or add an \enode to that \eclass
  (using the regular \texttt{add} and \texttt{merge} methods of the \egraph).

As mentioned above, other equality saturation implementations have implemented
  constant folding as custom, ad hoc passes over the \egraph.
We can formulate constant folding as an \eclass analysis that highlights the
  parallels with abstract interpretation.
Let the domain $D = \texttt{Option<Constant>}$, and let the \texttt{join}
  operation be the ``\texttt{or}'' operation of the \texttt{Option} type:
\\
\begin{minipage}{\linewidth}
\begin{lstlisting}[language=Rust, basicstyle=\ttfamily\footnotesize, xleftmargin=35mm]
match (a, b) {
  (None,    None   ) => None,
  (Some(x), None   ) => Some(x),
  (None,    Some(y)) => Some(y),
  (Some(x), Some(y)) => { assert!(x == y); Some(x) }
}
\end{lstlisting}
\end{minipage}
\\
Note how \textsf{join} can also aid in debugging by checking properties about
  values that are unified in the \egraph;
  in this case we assert that all terms represented in an \eclass should have
  the same constant value.
The \textsf{make} operation serves as the abstraction function, returning the
  constant value of an \enode if it can be computed from the constant values
  associated with its children \eclasses.
The \textsf{modify} operation serves as a concretization function in this
  setting.
If $d_{c}$ is a constant value, then $\textsf{modify}(c)$ would add
  $\gamma(d_{c}) = n$ to $c$, where $\gamma$ concretizes the constant value into
  a childless \enode.

Constant folding is an admittedly simple analysis, but one that did not formerly
  fit within the equality saturation framework.
\Eclass analyses support more complicated analyses in a general way, as
  discussed in later sections on the \egg implementation and case studies
  (Sections \ref{sec:impl} and \ref{sec:case-studies}).

\subsection{Conditional and Dynamic Rewrites}
\label{sec:rewrites}

In equality saturation applications, most of the rewrites are purely
  syntactic.
In some cases, additional data may be needed to determine if or how to perform
  the rewrite.
For example, the rewrite $x / x \to 1$ is only valid if $x \neq 0$.
A more complex rewrite may need to compute the right-hand side dynamically based
  on an analysis fact from the left-hand side.

The right-hand side of a rewrite can be generalized to a function
  \textsf{apply} that takes a substitution and an \eclass generated from
  e-matching the left-hand side, and produces a term to be added to the \egraph
  and unified with the matched \eclass.
For a purely syntactic rewrite, the \textsf{apply} function need not inspect the
  matched \eclass in any way; it would simply apply
  the substitution to the right-hand pattern to produce a new term.

\Eclass analyses greatly increase the utility of this generalized form of
  rewriting.
The \textsf{apply} function can look at the analysis data for the matched
  \eclass or any of the \eclasses in the substitution to determine if or how to
  construct the right-hand side term.
These kinds of rewrites can broken down further into two categories:
\begin{itemize}
  \item \textit{Conditional} rewrites like $x / x \to 1$ that are purely
  syntactic but whose validity depends on checking some analysis data;
  \item \textit{Dynamic} rewrites that compute the right-hand side based on
  analysis data.
\end{itemize}

Conditional rewrites are a subset of the more general dynamic rewrites.
Our \egg implementation supports both.
The example in \autoref{sec:impl} and case studies in \autoref{sec:case-studies}
  heavily use generalized rewrites, as it is typically the most convenient way
  to incorporate domain knowledge into the equality saturation
  framework.

\subsection{Extraction}
\label{sec:tricks-extraction}

Equality saturation typically ends with an extraction phase that selects an
  optimal represented term from an \eclass according to some cost function.
In many domains \cite{herbie, szalinski}, AST size
  (sometimes weighted differently for different operators) suffices as a simple,
  local cost function.
We say a cost function $k$ is local when the cost of a term $f(a_{1}, ...)$ can be
  computed from the function symbol $f$ and the costs of the children.
With such cost functions, extracting an optimal term can be efficiently done
  with a fixed-point traversal over the \egraph that selects the minimum cost
  \enode from each \eclass \cite{herbie}.

Extraction can be formulated as an \eclass analysis when the cost function
  is local.
The analysis data is a tuple $(n, k(n))$ where $n$ is the cheapest \enode
  in that \eclass and $k(n)$ its cost.
The $\textsf{make}(n)$ operation calculates the cost $k(n)$ based on
  the analysis data (which contain the minimum costs) of $n$'s children.
The \textsf{merge} operation simply takes the tuple with lower cost.
The semilattice portion of the analysis invariant then guarantees that the
  analysis data will contain the lowest-cost \enode in each class.
Extract can then proceed recursively;
  if the analysis data for \eclass $c$ gives $f(c_{1}, c_{2}, ...)$ as the optimal \enode,
  the optimal term represented in $c$ is
  $\textsf{extract}(c) = f( \textsf{extract}(c_{1}), \textsf{extract}(c_{2}), ... )$.
This not only further demonstrates the generality of \eclass analyses, but also
  provides the ability to do extraction ``on the fly''; conditional and dynamic
  rewrites can determine their behavior based on the cheapest term in an \eclass.

Extraction (whether done as a separate pass or an \eclass analysis) can also
  benefit from the analysis data.
Typically, a local cost function can only look at the function symbol of the
  \enode $n$ and the costs of $n$'s children.
When an \eclass analysis is attached to the \egraph, however, a cost function
  may observe the data associated with $n$'s \eclass, as well as the data
  associated with $n$'s children.
This allows a cost function to depend on computed facts rather that just purely
  syntactic information.
In other words, the cost of an operator may differ based on its inputs.
\autoref{sec:spores} provides a motivating case study wherein an \eclass
  analysis computes the size and shape of tensors, and this size information
  informs the cost function.


\section{\egg: Easy, Extensible, and Efficient \Egraphs}

\label{sec:egg}
\label{sec:impl}
\label{sec:lambda}

We implemented the techniques of rebuilding and \eclass analysis in \egg,
  an easy-to-use, extensible, and efficient \egraph library.
To the best of our knowledge,
  \egg is the first general-purpose, reusable \egraph implementation.
This has allowed focused effort on ease of use and optimization,
  knowing that any benefits will
  be seen across use cases as opposed to a single, ad hoc instance.

This section details \egg's implementation and some of the various
  optimizations and tools it provides to the user.
We use an extended example of a partial evaluator for the lambda calculus\footnote{
  \Egraphs do not have any ``built-in'' support for binding;
  for example, equality modulo alpha renaming is not free.
  The explicit substitution provided in this section is is illustrative but rather high in performance cost.
  Better support for languages with binding is important future work.
},
  for which we provide the complete source code (which few changes for readability)
  in \autoref{fig:lambda-lang} and \autoref{fig:lambda-analysis}.
While contrived, this example is compact and familiar, and it highlights
  (1) how \egg is used and (2) some of its novel features like
  \eclass analyses and dynamic rewrites.
It demonstrates how \egg can tackle binding,
  a perennially tough problem for \egraphs,
  with a simple explicit substitution approach
  powered by \egg's extensibility.
\autoref{sec:case-studies} goes further, providing real-world case studies of
  published projects that have depended on \egg.

\egg is implemented in \textasciitilde{}5000 lines of Rust,\footnote
{
  \citeauthor{rust} is a high-level systems programming language.
  \egg has been integrated into applications written in other
  programming languages using both C FFI and serialization approaches.
}
including code, tests, and documentation.
\egg is open-source, well-documented, and distributed via Rust's package
  management system.\footnote{
  Source: \url{https://github.com/mwillsey/egg}.
  Documentation: \url{https://docs.rs/egg}.
  Package: \url{https://crates.io/crates/egg}.
}
All of \egg's components are generic over the
  user-provided language, analysis, and cost functions.

\subsection{Ease of Use}
\label{sec:egg-easy}

\begin{figure}
\begin{subfigure}[t]{0.48\linewidth}
  \begin{lstlisting}[language=Rust, basicstyle=\tiny\ttfamily, numbers=left, escapechar=|]
define_language! {
  enum Lambda {
    // enum variants have data or children (eclass Ids)
    // [Id; N] is an array of N `Id`s

    // base type operators
    "+" = Add([Id; 2]), "=" = Eq([Id; 2]),
    "if" = If([Id; 3]),

    // functions and binding
    "app" = App([Id; 2]), "lam" = Lambda([Id; 2]),
    "let" = Let([Id; 3]), "fix" = Fix([Id; 2]),

    // (var x) is a use of `x` as an expression
    "var" = Use(Id),
    // (subst a x b) substitutes a for (var x) in b
    "subst" = Subst([Id; 3]),

    // base types have no children, only data
    Bool(bool), Num(i32), Symbol(String),
  }
}

// example terms and what they simplify to
// pulled directly from the |\egg|test suite

test_fn! { lambda_under, rules(),
  "(lam x (+ 4 (app (lam y (var y)) 4)))"
  => "(lam x 8))",
}

test_fn! { lambda_compose_many, rules(),
  "(let compose (lam f (lam g (lam x
                (app (var f)
                     (app (var g) (var x))))))
   (let add1 (lam y (+ (var y) 1))
   (app (app (var compose) (var add1))
        (app (app (var compose) (var add1))
             (app (app (var compose) (var add1))
                  (app (app (var compose) (var add1))
                       (var add1)))))))"
  => "(lam ?x (+ (var ?x) 5))"
}

test_fn! { lambda_if_elim, rules(),
  "(if (= (var a) (var b))
       (+ (var a) (var a))
       (+ (var a) (var b)))"
  => "(+ (var a) (var b))"
}\end{lstlisting}
\end{subfigure}
\hfill
\begin{subfigure}[t]{0.48\linewidth}
  \begin{lstlisting}[language=Rust, basicstyle=\tiny\ttfamily, escapechar=|, numbers=left, firstnumber=51]
// Returns a list of rewrite rules
fn rules() -> Vec<Rewrite<Lambda, LambdaAnalysis>> { vec![

 // open term rules
 rw!("if-true";  "(if  true ?then ?else)" => "?then"),
 rw!("if-false"; "(if false ?then ?else)" => "?else"),
 rw!("if-elim";  "(if (= (var ?x) ?e) ?then ?else)" => "?else"
     if ConditionEqual::parse("(let ?x ?e ?then)",
                              "(let ?x ?e ?else)")),
 rw!("add-comm";  "(+ ?a ?b)"        => "(+ ?b ?a)"),
 rw!("add-assoc"; "(+ (+ ?a ?b) ?c)" => "(+ ?a (+ ?b ?c))"),
 rw!("eq-comm";   "(= ?a ?b)"        => "(= ?b ?a)"),

 // substitution introduction
 rw!("fix";     "(fix ?v ?e)" =>
                "(let ?v (fix ?v ?e) ?e)"),
 rw!("beta";    "(app (lam ?v ?body) ?e)" =>
                "(let ?v ?e ?body)"),

 // substitution propagation
 rw!("let-app"; "(let ?v ?e (app ?a ?b))" =>
                "(app (let ?v ?e ?a) (let ?v ?e ?b))"),
 rw!("let-add"; "(let ?v ?e (+   ?a ?b))" =>
                "(+   (let ?v ?e ?a) (let ?v ?e ?b))"),
 rw!("let-eq";  "(let ?v ?e (=   ?a ?b))" =>
                "(=   (let ?v ?e ?a) (let ?v ?e ?b))"),
 rw!("let-if";  "(let ?v ?e (if ?cond ?then ?else))" =>
                "(if (let ?v ?e ?cond)
                     (let ?v ?e ?then)
                     (let ?v ?e ?else))"),

 // substitution elimination
 rw!("let-const";    "(let ?v ?e ?c)" => "?c"
     if is_const(var("?c"))),
 rw!("let-var-same"; "(let ?v1 ?e (var ?v1))" => "?e"),
 rw!("let-var-diff"; "(let ?v1 ?e (var ?v2))" => "(var ?v2)"
     if is_not_same_var(var("?v1"), var("?v2"))),
 rw!("let-lam-same"; "(let ?v1 ?e (lam ?v1 ?body))" =>
                     "(lam ?v1 ?body)"),
 rw!("let-lam-diff"; "(let ?v1 ?e (lam ?v2 ?body))" =>
     ( CaptureAvoid {
        fresh: var("?fresh"), v2: var("?v2"), e: var("?e"),
        if_not_free: "(lam ?v2 (let ?v1 ?e ?body))"
                     .parse().unwrap(),
        if_free: "(lam ?fresh (let ?v1 ?e
                              (let ?v2 (var ?fresh) ?body)))"
                 .parse().unwrap(),
     })
     if is_not_same_var(var("?v1"), var("?v2"))),
]}\end{lstlisting}
\end{subfigure}
\caption[Language and rewrites for the lambda calculus in \egg]{
\egg is generic over user-defined languages;
  here we define a language and rewrite rules for a lambda calculus partial evaluator.
The provided \texttt{define\_language!} macro (lines 1-22) allows the simple definition
  of a language as a Rust \texttt{enum}, automatically deriving parsing and
  pretty printing.
A value of type \texttt{Lambda} is an \enode that holds either data that the
  user can inspect or some number of \eclass children (\eclass \texttt{Id}s).

Rewrite rules can also be defined succinctly (lines 51-100).
Patterns are parsed as s-expressions:
  strings from the \texttt{define\_language!} invocation (ex: \texttt{fix}, \texttt{=}, \texttt{+}) and
  data from the variants (ex: \texttt{false}, \texttt{1}) parse as operators or terms;
  names prefixed by ``\texttt{?}'' parse as pattern variables.

Some of the rewrites made are conditional using the
  ``\texttt{left => right if cond}''
  syntax.
The \texttt{if-elim} rewrite on line 57 uses \egg's provided
  \texttt{ConditionEqual} as a condition, only applying the right-hand side
  if the \egraph can prove the two argument patterns equivalent.
The final rewrite, \texttt{let-lam-diff}, is dynamic to support capture avoidance;
  the right-hand side is a Rust value that
  implements the \texttt{Applier} trait instead of a pattern.
\autoref{fig:lambda-analysis} contains the supporting code for these rewrites.

We also show some of the tests (lines 27-50)
  from \egg's \texttt{lambda} test suite.
The tests proceed by inserting the term on the left-hand side, running
  \egg's equality saturation, and then checking to make sure the right-hand
  pattern can be found in the same \eclass as the initial term.
}
\label{fig:lambda-rules}
\label{fig:lambda-lang}
\label{fig:lambda-examples}
\end{figure}


\egg's ease of use comes primarily from its design as a library.
By defining only a language and some rewrite rules,
  a user can quickly
  start developing a synthesis or optimization tool.
Using \egg as a Rust library,
  the user defines the language using the \texttt{define\_language!} macro
  shown in \autoref{fig:lambda-lang}, lines 1-22.
Childless variants in the language may contain data of user-defined types,
  and \eclass analyses or dynamic rewrites may inspect this data.


The user provides rewrites as shown in
  \autoref{fig:lambda-lang}, lines 51-100.
Each rewrite has a name, a left-hand side, and a right-hand side.
For purely syntactic rewrites, the right-hand is simply a pattern.
More complex rewrites can incorporate conditions or even dynamic right-hand
  sides, both explained in the \autoref{sec:egg-extensible} and \autoref{fig:lambda-applier}.

Equality saturation workflows, regardless of the application domain,
  typically have a similar structure:
add expressions to an empty \egraph, run rewrites until saturation or
  timeout, and extract the best equivalent expressions according to some cost
  function.
This ``outer loop'' of equality saturation involves a significant amount of
  error-prone boilerplate:
\begin{itemize}
  \item Checking for saturation, timeouts, and \egraph size limits.
  \item Orchestrating the read-phase, write-phase, rebuild system
    (\autoref{fig:rebuild-code}) that makes \egg fast.
  \item Recording performance data at each iteration.
  \item Potentially coordinating rule execution so that expansive rules like
    associativity do not dominate the \egraph.
  \item Finally, extracting the best expression(s) according to a
  user-defined cost function.
\end{itemize}

\egg provides these functionalities through its \texttt{Runner} and
  \texttt{Extractor} interfaces.
\texttt{Runner}s automatically detect saturation, and can be configured to stop
  after a time, \egraph size, or iterations limit.
The equality saturation loop provided by \egg calls \texttt{rebuild}, so users
  need not even know about \egg's deferred invariant maintenance.
\texttt{Runner}s record various metrics about each iteration automatically,
  and the user can hook into this to report relevant data.
\texttt{Extractor}s select the optimal term from an \egraph given a
  user-defined, local cost function.\footnote{
    As mentioned in \autoref{sec:tricks-extraction}, extraction can be
    implemented as part of an \eclass analysis.
    The separate \texttt{Extractor} feature is still useful for ergonomic and
    performance reasons.
  }
The two can be combined as well; users commonly record the ``best so far''
  expression by extracting in each iteration.

\autoref{fig:lambda-lang} also shows \egg's \texttt{test\_fn!}
  macro for easily creating tests (lines 27-50).
These tests create an \egraph with the given expression, run equality saturation
  using a \texttt{Runner}, and check to make sure the right-hand pattern can be
  found in the same \eclass as the initial expression.

\subsection{Extensibility}
\label{sec:egg-extensible}

For simple domains, defining a language and purely syntactic rewrites will
  suffice.
However, our partial evaluator requires interpreted reasoning, so we use some of
  \egg's more advanced features like \eclass analyses and dynamic rewrites.
Importantly, \egg supports these extensibility features as a library:
  the user need not modify the \egraph or \egg's internals.

\begin{figure}
\begin{minipage}[t]{0.49\linewidth}
  \begin{lstlisting}[language=Rust, basicstyle=\tiny\ttfamily, numbers=left]
type EGraph = egg::EGraph<Lambda, LambdaAnalysis>;
struct LambdaAnalysis;
struct FC {
  free: HashSet<Id>,    // our analysis data stores free vars
  constant: Option<Lambda>, // and the constant value, if any
}

// helper function to make pattern meta-variables
fn var(s: &str) -> Var { s.parse().unwrap() }

impl Analysis<Lambda> for LambdaAnalysis {
  type Data = FC; // attach an FC to each eclass
  // merge implements semilattice join by joining into `to`
  // returning true if the `to` data was modified
  fn merge(&self, to: &mut FC, from: FC) -> bool {
    let before_len = to.free.len();
    // union the free variables
    to.free.extend(from.free.iter().copied());
    if to.constant.is_none() && from.constant.is_some() {
      to.constant = from.constant;
      true
    } else {
      before_len != to.free.len()
    }
  }

  fn make(egraph: &EGraph, enode: &Lambda) -> FC {
    let f = |i: &Id| egraph[*i].data.free.iter().copied();
    let mut free = HashSet::default();
    match enode {
      Use(v) => { free.insert(*v); }
      Let([v, a, b]) => {
        free.extend(f(b)); free.remove(v); free.extend(f(a));
      }
      Lambda([v, b]) | Fix([v, b]) => {
        free.extend(f(b)); free.remove(v);
      }
      _ => enode.for_each_child(
             |c| free.extend(&egraph[c].data.free)),
    }
    FC { free: free, constant: eval(egraph, enode) }
  }

  fn modify(egraph: &mut EGraph, id: Id) {
    if let Some(c) = egraph[id].data.constant.clone() {
      let const_id = egraph.add(c);
      egraph.union(id, const_id);
    }
  }
}\end{lstlisting}
\end{minipage}
\hfill
\begin{minipage}[t]{0.46\linewidth}
  \begin{lstlisting}[language=Rust, basicstyle=\tiny\ttfamily, escapechar=@, numbers=left, firstnumber=51]
// evaluate an enode if the children have constants
// Rust's `?` extracts an Option, early returning if None
fn eval(eg: &EGraph, enode: &Lambda) -> Option<Lambda> {
  let c = |i: &Id| eg[*i].data.constant.clone();
  match enode {
    Num(_) | Bool(_) => Some(enode.clone()),
    Add([x, y]) => Some(Num(c(x)? + c(y)?)),
    Eq([x, y]) => Some(Bool(c(x)? == c(y)?)),
    _ => None,
  }
}

// Functions of this type can be conditions for rewrites
trait ConditionFn = Fn(&mut EGraph, Id, &Subst) -> bool;

// The following two functions return closures of the
// correct signature to be used as conditions in @\autoref{fig:lambda-rules}@.
fn is_not_same_var(v1: Var, v2: Var) -> impl ConditionFn {
    |eg, _, subst| eg.find(subst[v1]) != eg.find(subst[v2])
}
fn is_const(v: Var) -> impl ConditionFn {
     // check the LambdaAnalysis data
    |eg, _, subst| eg[subst[v]].data.constant.is_some()
}

struct CaptureAvoid {
  fresh: Var, v2: Var, e: Var,
  if_not_free: Pattern<Lambda>, if_free: Pattern<Lambda>,
}

impl Applier<Lambda, LambdaAnalysis> for CaptureAvoid {
  // Given the egraph, the matching eclass id, and the
  // substitution generated by the match, apply the rewrite
  fn apply_one(&self, egraph: &mut EGraph,
               id: Id, subst: &Subst) -> Vec<Id>
  {
    let (v2, e) = (subst[self.v2], subst[self.e]);
    let v2_free_in_e = egraph[e].data.free.contains(&v2);
    if v2_free_in_e {
      let mut subst = subst.clone();
      // make a fresh symbol using the eclass id
      let sym = Lambda::Symbol(format!("_{}", id).into());
      subst.insert(self.fresh, egraph.add(sym));
      // apply the given pattern with the modified subst
      self.if_free.apply_one(egraph, id, &subst)
    } else {
      self.if_not_free.apply_one(egraph, id, &subst)
    }
  }
}\end{lstlisting}
\end{minipage}
\caption[\Eclass analysis and conditional/dynamic rewrites for the lambda calculus]{
Our partial evaluator example highlights three important features \egg provides
  for extensibility: \eclass analyses, conditional rewrites, and dynamic
  rewrites.
  
The \texttt{LambdaAnalysis} type, which implements the \texttt{Analysis} trait,
  represents the \eclass analysis.
Its associated data (\texttt{FC}) stores
  the constant term from that \eclass (if any) and
  an over-approximation of the free variables used by terms in that \eclass.
The constant term is used to perform constant folding.
The \texttt{merge} operation implements the semilattice join, combining the free
  variable sets and taking a constant if one exists.
In \texttt{make}, the analysis computes the free variable sets based on the
  \enode and the free variables of its children;
  the \texttt{eval} generates the new constants if possible.
The \texttt{modify} hook of \texttt{Analysis} adds the constant to the \egraph.

Some of the conditional rewrites in \autoref{fig:lambda-rules} depend on
  conditions defined here.
Any function with the correct signature may serve as a condition.

The \texttt{CaptureAvoid} type implements the \texttt{Applier} trait, allowing
  it to serve as the right-hand side of a rewrite.
\texttt{CaptureAvoid} takes two patterns and some pattern variables.
It checks the free variable set to determine if a capture-avoiding substitution
  is required, applying the \texttt{if\_free} pattern if so and the
  \texttt{if\_not\_free} pattern otherwise.
}
\label{fig:lambda-applier}
\label{fig:lambda-analysis}
\end{figure}


\autoref{fig:lambda-applier} shows the remainder of the code for our lambda
  calculus partial evaluator.
It uses an \eclass analysis (\texttt{LambdaAnalysis})
  to track free variables and constants associated
  with each \eclass.
The implementation of the \eclass analysis is in Lines 11-50.
The \eclass analysis invariant
  guarantees that the analysis data contains an over-approximation of free variables
  from terms represented in that \eclass.
The analysis also does constant folding
  (see the \texttt{make} and \texttt{modify} methods).
The \texttt{let-lam-diff} rewrite (Line 90, \autoref{fig:lambda-rules})
  uses the \texttt{CaptureAvoid} (Lines 81-100, \autoref{fig:lambda-applier})
  dynamic right-hand side to do capture-avoiding
  substitution only when necessary based on the free variable information.
The conditional rewrites from \autoref{fig:lambda-rules} depend on the
  conditions \texttt{is\_not\_same\_var} and
  \texttt{is\_var} (Lines 68-74, \autoref{fig:lambda-applier})
  to ensure correct substitution.

\egg is extensible in other ways as well.
As mentioned above, \texttt{Extractor}s are parameterized by a user-provided
  cost function.
\texttt{Runner}s are also extensible with user-provided rule schedulers that can
  control the behavior of potentially troublesome rewrites.
\label{sec:rule-scheduling}
In typical equality saturation, each rewrite is searched for and applied each
  iteration.
This can cause certain rewrites, commonly associativity or distributivity,
  to dominate others and make the search space less productive.
Applied in moderation, these rewrites can trigger other rewrites and find
  greatly improved expressions,
  but they can also slow the search by
  exploding the \egraph exponentially in size.
By default, \egg uses the built-in backoff scheduler
  that identifies rewrites that are matching in exponentially-growing
  locations and temporarily bans them.
We have observed that this greatly reduced run time (producing the same results)
  in many settings.
\egg can also use a conventional every-rule-every-time scheduler, or the user
  can supply their own.

\subsection{Efficiency}
\label{sec:egg-efficient}

\egg's novel \textit{rebuilding} algorithm (\autoref{sec:rebuild})
combined with systems programming best practices
  makes \egraphs---and the equality saturation
  use case in particular---more efficient than prior tools.

\egg is implemented in Rust, giving the compiler freedom to
  specialize and inline user-written code.
This is especially important as
  \egg's generic nature leads to tight interaction
  between library code
  (e.g., searching for rewrites) and user code (e.g., comparing operators).
\egg is designed from the ground up to use cache-friendly,
  flat buffers with minimal indirection for most internal data structures.
This is in sharp contrast to traditional representations of \egraphs
  \cite{nelson, simplify} that contains many tree- and linked list-like data
  structures.
\egg additionally compiles patterns to be executed by a small virtual machine
  \cite{ematching}, as opposed to recursively walking the tree-like
  representation of patterns.

Aside from deferred rebuilding, \egg's equality saturation algorithm leads to
  implementation-level performance enhancements.
Searching for rewrite matches, which is the bulk of running time, can be
  parallelized thanks to the phase separation.
Either the rules or \eclasses could be searched in parallel.
Furthermore, the once-per-iteration frequency of rebuilding allows \egg to
  establish other performance-enhancing invariants that hold during the
  read-only search phase.
For example, \egg sorts \enodes within each \eclass to enable binary search, and
  also maintains a cache mapping function symbols to \eclasses that
  contain \enodes with that function symbol.

Many of \egg's extensibility features can also be used to improve performance.
As mentioned above, rule scheduling can lead to great performance improvement in
  the face of ``expansive'' rules that would otherwise dominate the search
  space.
The \texttt{Runner} interface also supports user hooks that can stop
  the equality saturation after some arbitrary condition.
This can be useful when using equality saturation to prove terms equal; once
  they are unified, there is no point in continuing.
\label{sec:egg-batched}
\egg's \texttt{Runner}s also support batch simplification, where multiple terms
  can be added to the initial \egraph before running equality saturation.
If the terms are substantially similar, both rewriting and any \eclass analyses
  will benefit from the \egraph's inherent structural deduplication.
The case study in \autoref{sec:herbie} uses batch simplification to achieve
  a large speedup with simplifying similar expressions.


\section{Case Studies}
\label{sec:case-studies}

This section relates three independently-developed, published projects from diverse domains
  that incorporated \egg
  as an easy-to-use, high-performance \egraph implementation.
In all three cases, the developers had first rolled their own \egraph
  implementations.
\Egg allowed them to delete code, gain performance, and in some cases
  dramatically broaden the project's scope thanks to \egg's speed and
  flexibility.
In addition to gaining performance, all three projects use \egg's novel
  extensibility features like \eclass analyses and dynamic/conditional rewrites.

\subsection{Herbie: Improving Floating Point Accuracy}
\label{sec:herbie}

Herbie automatically improves accuracy
  for floating-point expressions,
  using random sampling to measure error,
  a set of rewrite rules for generating program variants,
  and algorithms that prune and combine program variants
  to achieve minimal error.
Herbie received PLDI 2015's Distinguished Paper award~\cite{herbie}
  and has been continuously developed since then,
  sporting hundreds of Github stars, hundreds of downloads,
  and thousands of users on its online version.
Herbie uses \egraphs for algebraic simplification of mathematical expressions,
  which is especially important for avoiding floating-point errors
  introduced by cancellation, function inverses, and redundant computation.

Until our case study,
  Herbie used a custom \egraph implementation
  written in Racket (Herbie's implementation language)
  that closely followed traditional \egraph implementations.
With timeouts disabled,
  \egraph-based simplification consumed
  the vast majority of Herbie's run time.
As a fix, Herbie sharply limits the simplification process,
  placing a size limit on the \egraph itself and a time limit on the whole
  procedure.
When the timeout is exceeded, simplification fails altogether.
Furthermore, the Herbie authors knew of several features
  that they believed would improve Herbie's output
  but could not be implemented because
  they required more calls to simplification
  and would thus introduce unacceptable slowdowns.
Taken together, slow simplification reduced Herbie's performance, completeness,
  and efficacy.

We implemented a \egg simplification backend for Herbie.
The \egg backend is over $3000\times$ faster than Herbie's initial simplifier and
  is now used by default as of Herbie 1.4.
Herbie has also backported some of \egg's features like batch simplification and
  rebuilding to its \egraph implementation
  (which is still usable, just not the default),
  demonstrating the portability of \egg's conceptual improvements.

\subsubsection{Implementation}

Herbie is implemented in Racket while \egg is in Rust;
  the \egg simplification backend is thus implemented as a Rust library that
  provides a C-level API for Herbie to access via foreign-function interface (FFI).
The Rust library defines the Herbie expression grammar
  (with named constants, numeric constants, variables, and operations)
  as well as the \eclass analysis necessary to do constant folding.
The library is implemented in under 500 lines of Rust.

Herbie's set of rewrite rules is not fixed;
  users can select which rewrites to use using command-line flags.
Herbie serializes the rewrites to strings,
  and the \egg backend parses and instantiates them on the Rust side.

Herbie separates exact and inexact program constants:
  exact operations on exact constants
  (such as the addition of two rational numbers)
  are evaluated and added to the \egraph,
  while operations on inexact constants or that yield inexact outputs
  are not.
We thus split numeric constants in the Rust-side grammar
  between exact rational numbers and inexact constants,
  which are described by an opaque identifier,
  and transformed Racket-side expressions into this form
  before serializing them and passing them to the Rust driver.
To evaluate operations on exact constants,
  we used the constant folding \eclass analysis
  to track the ``exact value'' of each \eclass.%
\footnote{Herbie's rewrite rules guarantee that different exact values
  can never become equal; the semilattice \textsf{join} checks this invariant on the Rust side.}
Every time an operation \enode is added to the \egg \egraph,
  we check whether all arguments to that operation have exact value (using the analysis data),
  and if so do rational number arithmetic to evaluate it.
The \eclass analysis is cleaner than the corresponding code in Herbie's implementation,
  which is a built-in pass over the entire \egraph.

\subsubsection{Results}

\begin{figure}
  \centering
  \includegraphics[height=5.5cm]{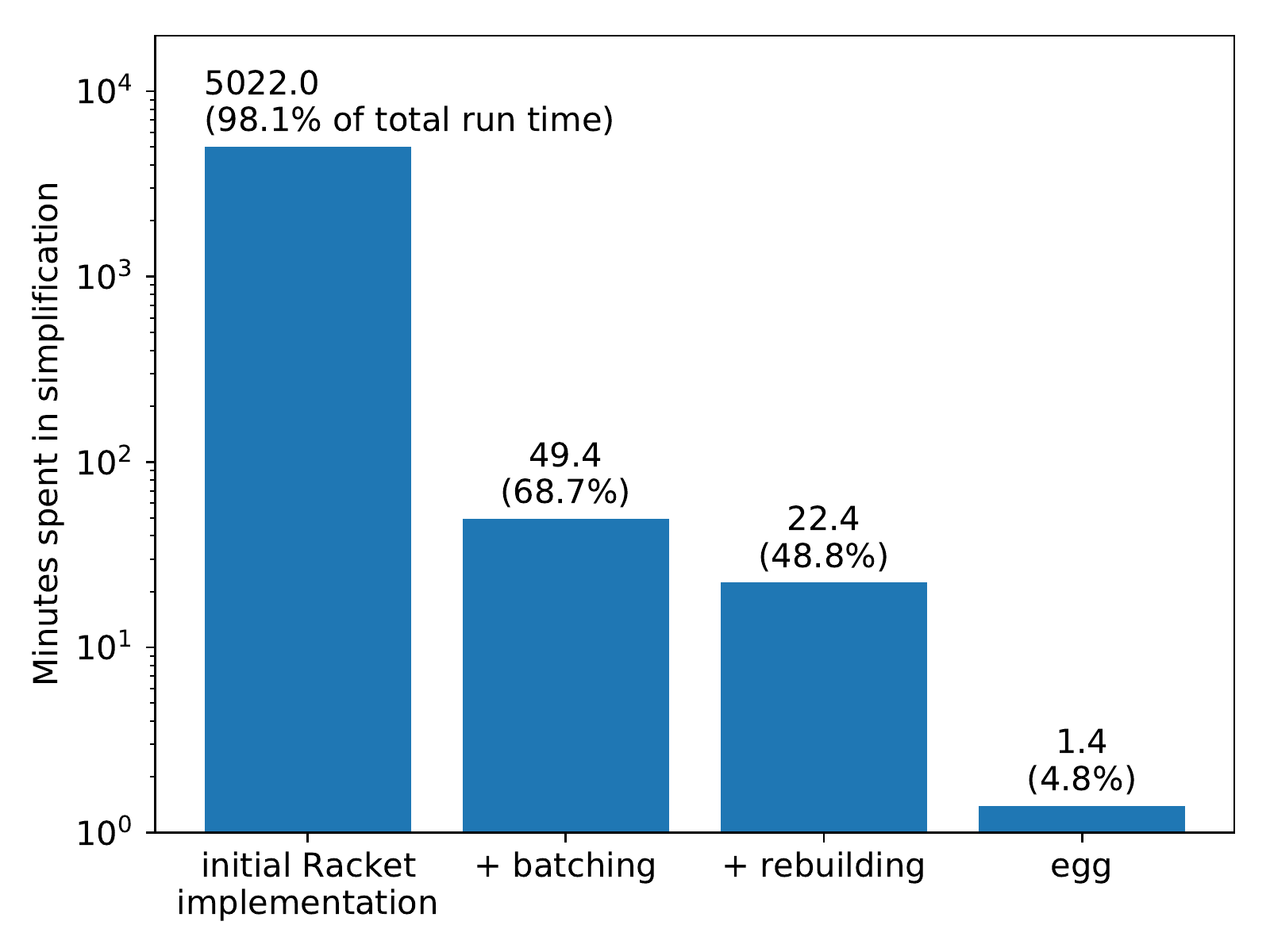}
  \caption{
    Herbie sped up its expression simplification phase
      by adopting \egg-inspired features like
      batched simplification and rebuilding
      into its Racket-based \egraph implementation.
    Herbie also supports using \egg itself for additional speedup.
    Note that the y-axis is log-scale.
  }
  \label{fig:herbie-results}
\end{figure}

Our \egg simplification backend
  is a drop-in replacement to the existing Herbie simplifier,
  making it easy to compare speed and results.
We compare using
  Herbie's standard test suite of roughly 500 benchmarks,
  with timeouts disabled.
\autoref{fig:herbie-results} shows the results.
The \egg simplification backend is over
  $3000\times$ faster than Herbie's initial simplifier.
This speedup eliminated Herbie's largest bottleneck:
  the initial implementation dominated Herbie's total run time at $98.1\%$,
  backporting \egg improvements into Herbie cuts
  that to about half the total run time,
  and \egg simplification takes under $5\%$ of the total run time.
Practically, the run time of Herbie's initial implementation was smaller, since
  timeouts cause tests failures when simplification takes too long.
Therefore, the speedup also improved Herbie's completeness,
  as simplification now never times out.

Since incorporating \egg into Herbie, the Herbie developers have backported some
  of \egg's key performance improvements into the Racket \egraph implementation.
First, batch simplification gives a large speedup because Herbie simplifies many
  similar expressions.
When done simultaneously in one equality saturation, the \egraph's structural
  sharing can massively deduplicate work.
Second, deferring rebuilding (as discussed in \autoref{sec:rebuilding}) gives a
  further $2.2\times$ speedup.
As demonstrated in \autoref{fig:eval-iter}, rebuilding offers an asymptotic
  speedup, so Herbie's improved implementation (and the \egg backend as well)
  will scale better as the search size grows.

\subsection{Spores: Optimizing Linear Algebra}
\label{sec:spores}

Spores \cite{spores} is an optimizer for machine learning programs. It
translates linear algebra (LA) expressions to relational algebra (RA), performs
rewrites, and finally translates the result back to linear algebra. Each rewrite
is built up from simple identities in relational algebra like the associativity
of join. These relational identities express more fine-grained equality than
textbook linear algebra identities, allowing Spores to discover novel
optimizations not found by traditional optimizers based on LA identities. Spores
performs holistic optimization, taking into account the complex interactions
among factors like sparsity, common subexpressions, and fusible operators and
their impact on execution time.


\subsubsection{Implementation}

Spores is implemented entirely in Rust using \Egg.
\Egg empowers Spores to orchestrate
  the complex interactions described above elegantly and effortlessly.
Spores works in three steps: first, it translates the input LA expression to RA;
second, it optimizes the RA expression by equality saturation; finally, it
translates the optimized RA expression back to LA.
Since the translation between LA and RA is
straightforward, we focus the discussion on the equality saturation step in RA.
Spores represents a relation as a function from tuples to real numbers: $A:
(a_1, a_2, ..., a_n) \rightarrow \mathbb{R}$. This is similar to the index
notation in linear algebra, where a matrix A can be viewed as a function
$\lambda i, j . A_{ij}$. A tuple is identified with a named record,
e.g. $(1, 2) = \{a_1: 1, a_2: 2\} = \{a_2: 2, a_1: 1\}$, so that order in a
tuple doesn't matter. There are just three operations on relations: join, union
and aggregate. Join ($\otimes$) takes two relations and returns their natural
join, multiplying the associated real number for joined tuples:
\[ A \otimes B = \lambda \bar{a} \cup \bar{b} . A(\bar{a}) \times B(\bar{b})\]
Here $\bar{a}$ is the set of field names for the records in $A$. In RA
terminology, $\bar{a}$ is the {\em schema} of $A$. Union ($\oplus$) is a join in
disguise: it also performs natural join on its two arguments, but adds the
associated real instead of multiplying it:
\[A \oplus B = \lambda \bar{a} \cup \bar{b} . A(\bar{a}) + B(\bar{b})\] Finally,
aggregate ($\Sigma$) sums its argument along a given dimension. It
coincides precisely with the ``sigma notation'' in mathematics:
\[\sum_{a_i} A = \lambda \bar{a}-a_i . \sum_{a_i} A(\bar{a})\]

\begin{figure}
\centering
\begin{align}
  \label{eq:RRC_pp} A \oplus (B \oplus C) &= \oplus (A, B, C)
  & \text{($\oplus$ is assoc. \& comm.)}
  \\
  \label{eq:RRC_mm} A \otimes (B \otimes C) &= \otimes (A, B, C)
  & \text{($\otimes$ is assoc. \& comm.)}
  \\
  \label{eq:RRC_mp} A \otimes (B \oplus C) &= A \otimes B \oplus A \otimes C
  & \text{($\otimes$ distributes over $\oplus$)}
  \\
  \label{eq:RRC_ap} \sum_i (A \oplus B) &= \sum_i A \oplus \sum_i B
  \\
  \label{eq:RRC_aa} \sum_i \sum_j A &= \sum_{i,j} A
  \\
  \label{eq:RRC_ma} A \otimes \sum_i B &= \sum_i (A \otimes B)
  &\text{(requires $i \not\in A$)}
  \\
  \label{eq:RRC_ac} \sum_i A &= A \otimes \textsf{dimension}(i)
  &\text{(requires $i \not\in A$)}
\end{align}
\caption{RA equality rules $R_{EQ}$.}
\label{fig:RRC}
\end{figure}

The RA identities, presented in \autoref{fig:RRC}, are also simple and intuitive.
The notation $i \not\in A$ means $i$ is not in the schema of $A$, and $dim(i)$
is the size of dimension $i$ (e.g. length of rows in a matrix). In
\autoref{eq:RRC_ma}, when $i \in A$, we first rename every $i$ to a fresh variable
$i'$ in $B$, which gives us: $A \otimes \sum_i B = \sum_{i'} (A \otimes B[i
\rightarrow i'])$. In addition to these equalities, Spores also supports
replacing expressions with fused operators. For example, $(X-UV)^2$ can be
replaced by $sqloss(X, U, V)$ which streams values from $X, U, V$ and
computes the result without creating intermediate matrices. Each of these fused
operators is encoded with a simple identity in \Egg.

Note that \autoref{eq:RRC_ma} requires a way
to store the schema of every expression during optimization.
Spores uses an \eclass analysis to annotate \eclasses with the appropriate
  schema. It also leverages the \eclass analysis for cost estimation,
  using a conservative cost model that overapproximates.
As a result, equivalent expressions may have different cost estimates.
The \texttt{merge} operation on the analysis data takes the lower cost,
  incrementally improving the cost estimate.
Finally, Spores' \eclass analysis also performs constant folding.
As a whole, the \eclass analysis is a composition of three smaller analyses
  in a similar style to the composition of lattices in abstract interpretation.

\subsubsection{Results}
Spores is integrated into Apache SystemML \cite{Boehm_2019} in a prototype,
where it is able to derive all of 84 hand-written rules and heuristics for
sum-product optimization. It also discovered novel rewrites that contribute
to $1.2\times$ to $5\times$ speedup in end-to-end experiments. With greedy
extraction, all compilations completed within a second.

\subsection{Szalinski: Decompiling CAD into Structured Programs}
\label{sec:szalinski}

\begin{figure}
  \centering
  \includegraphics[width=\linewidth]{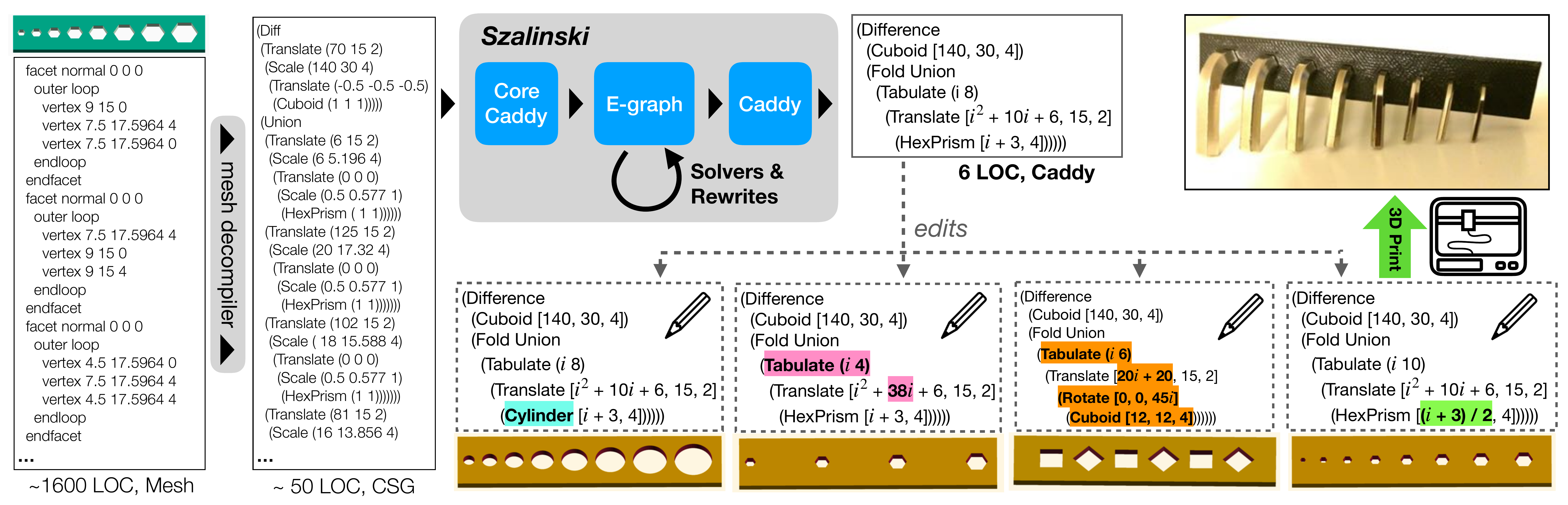}
  \caption[Szalinski decompiles flat CSG into structured CAD]{
  (Figure from Nandi et.\ al.\ \cite{szalinski})
  Existing mesh decompilers turn
    triangle meshes into flat, computational solid geometry (CSG) expressions.
  \sz~\cite{szalinski} takes in these CSG expressions
    in a format called Core Caddy,
    and it synthesizes smaller, structured programs in language called Caddy
    that is enriched with functional-style features.
  This can ease customization by simplifying edits:
    small, mostly local changes
    yield usefully different models.
  The photo shows the 3D printed hex wrench holder after
    customizing hole sizes.
  \sz is powered by \egg's extensible equality saturation, relying on its high
    performance, \eclass analyses, and dynamic rewrites.
  }
  \label{fig:sz-overview}
\end{figure}

Several tools have emerged
  that reverse engineer high level
  Computer Aided Design (CAD) models from polygon
  meshes and voxels~\cite{reincarnate, inverse, shape, csgnet, latex}.
The output of these tools are constructive solid geometry (CSG) programs.
A CSG program is comprised of
  3D solids like cubes, spheres, cylinders,
  affine transformations like scale, translate, rotate
  (which take a 3D vector and a CSG expression as arguments),
  and binary operators like union, intersection, and difference
  that combine CSG expressions.
For repetitive models like a gear, CSG programs can be too long
  and therefore difficult to comprehend.
A recent tool, \sz~\cite{szalinski},
  extracts the inherent structure
  in the CSG outputs of mesh decompilation tools
  by automatically inferring maps and folds (\autoref{fig:sz-overview}).
\sz accomplished this using \egg's extensible equality saturation system,
  allowing it to:

\begin{itemize}

  \item Discover structure using loop rerolling rules.
    This allows \sz to infer functional patterns like
    \texttt{Fold}, \texttt{Map2}, \texttt{Repeat} and
    \texttt{Tabulate} from flat CSG inputs.

  \item Identify equivalence among CAD terms that are
    expressed as different expressions by mesh decompilers.
    \sz accomplishes this by using CAD identities.
    An example of one such CAD identity in \sz is
    $e \leftrightarrow \mathit{rotate}~[0 ~ 0 ~ 0] ~ e$.
    This implies that any CAD expression $e$
    is equivalent to a CAD expression that applies
    a rotation by zero degrees about x, y, and z axes
    to $e$.

  \item Use external solvers to
    speculatively add potentially profitable
    expressions to the \egraph.
    Mesh decompilers often generate CSG expressions
    that order and/or group list elements in
    non-intuitive ways.
    To recover structure from such expressions,
    a tool like \sz must be able to reorder and regroup
    lists that expose any latent structure.

\end{itemize}

\subsubsection{Implementation}

Even though CAD is
  different from traditional languages
  targeted by programming language techniques,
  \egg supports \sz's CAD language in a straightforward manner.
\sz uses purely syntactic rewrites to express
  CAD identities and some loop rerolling rules
  (like inferring a \texttt{Fold} from a list of CAD expressions).
Critically, however, \sz relies on \egg's
  dynamic rewrites and \eclass analysis to infer functions
  for lists.

Consider the flat CSG program in \autoref{fig:sz-egg-input}.
A structure finding rewrite first rewrites the flat list of \texttt{Union}s to:
$$\texttt{(Fold Union (Map2 Translate [(0 0 0) (2 0 0) ...] (Repeat Cube 5)))}$$
The list of vectors is stored as \texttt{Cons} elements (sugared above for brevity).
\sz uses an \eclass analysis to track the accumulated lists in a similar style
  to constant folding.
Then, a dynamic rewrite uses an arithmetic solver to rewrite the concrete
  list of 3D vectors in the analysis data
  to \mbox{\texttt{(Tabulate (i 5) (* 2 i))}}.
A final set of syntactic rewrites can hoist the \texttt{Tabulate}, yielding the
  result on the right of \autoref{fig:sz-egg}.
Thanks to the set of syntactic CAD rewrites, this structure finding even works
  in the face of CAD identities.
For example, the original program may omit the no-op
  \texttt{Translate (0 0 0)}, even though it is necessary to see repetitive
  structure.

\begin{figure}
\begin{subfigure}[b]{0.3\linewidth}
  \includegraphics[width=\linewidth]{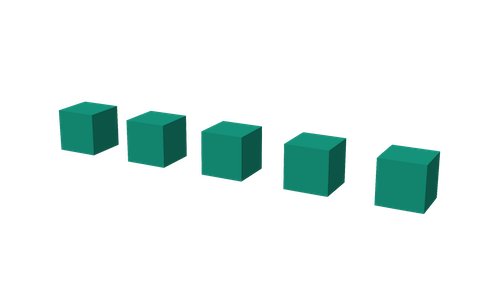}
  \caption{Five cubes in a line.}
\end{subfigure}
\hfill
\begin{subfigure}[b]{0.3\linewidth}
  \begin{lstlisting}[language=Rust, gobble=4, basicstyle=\footnotesize\ttfamily]
    (Union
      (Translate (0 0 0) Cube)
      (Translate (2 0 0) Cube)
      (Translate (4 0 0) Cube)
      (Translate (6 0 0) Cube)
      (Translate (8 0 0) Cube))
  \end{lstlisting}
  \caption{Flat CSG input to \sz.}
  \label{fig:sz-egg-input}
\end{subfigure}
\hfill
\begin{subfigure}[b]{0.35\linewidth}
  \begin{lstlisting}[language=Rust, gobble=4, basicstyle=\footnotesize\ttfamily, showlines=true, xleftmargin=5mm]
    (Fold Union
      (Tabulate (i 5)
        (Translate
          ((* 2 i) 0 0)
          Cube)))

  \end{lstlisting}
  \caption{Output captures the repetition.}
  \label{fig:sz-egg-output}
\end{subfigure}
  \caption{
    \sz integrates solvers into \egg's equality saturation as a dynamic rewrite.
    The solver-backed rewrites can transform repetitive lists into
    \texttt{Tabulate} expressions that capture the repetitive structure.
  }
  \label{fig:sz-egg}
\end{figure}

In many cases, the repetitive structure of input CSG expression is further
  obfuscated because subexpressions may appear in arbitrary order.
For these inputs, the arithmetic solvers must first reorder
  the expressions to find a closed form like a \texttt{Tabulate}
  as shown in \autoref{fig:sz-egg}.
However, reordering a list does not preserve equivalence, so adding it to the
  \eclass of the concrete list would be unsound.
\sz therefore introduces \textit{inverse transformations},
  a novel technique that allows solvers to speculatively reorder and regroup list
  elements to find a closed form.
The solvers annotate the potentially profitable expression with the
  permutation or grouping that led to the successful discovery
  of the closed form.
Later in the rewriting process, syntactic rewrites eliminate the inverse
  transformations when possible
  (e.g., reordering lists under a \texttt{Fold Union} can be eliminated).
\egg supported this novel technique without modification.

\subsubsection{Results}
\sz's initial protoype used a custom \egraph written in OCaml.
Anecdotally, switching to \egg
  removed most of the code,
  eliminated bugs,
  facilitated the key contributions of solver-backed rewrites and inverse transformations,
  and made the tool about $1000 \times$ faster.
\egg's performance allowed a shift from running on small, hand-picked
  examples to a comprehensive evaluation on over 2000 real-world models from
  a 3D model sharing forum~\cite{szalinski}.


\section{Related Work}
\label{sec:related}

\paragraph{Term Rewriting}
Term rewriting~\cite{rewritesystems} has been used widely to facilitate
equational reasoning for program
optimizations~\cite{DBLP:conf/scitools/BoyleHW96,
  DBLP:journals/toplas/BrandHKO02, DBLP:conf/icfp/VisserBT98}. A term rewriting
system applies a database of semantics preserving rewrites or axioms to an input
expression to get a new expression, which may, according to some cost function,
be more profitable compared to the input. Rewrites are typically symbolic and
have a left hand side and a right hand side. To apply a rewrite to an
expression, a rewrite system implements pattern matching---if the left hand side
of a rewrite rule matches with the input expression, the system computes a
substitution which is then applied to the right-hand side of the rewrite rule.
Upon applying a rewrite rule, a rewrite system typically replaces the old
expression by the new expression. This can lead to the \textit{phase ordering}
problem--- it makes it impossible to apply a rewrite to the old expression in
the future which could have led to a more optimal result.

\paragraph{\Egraphs and E-matching}
\Egraph were originally proposed several decades ago as an
efficient data structure for maintaining congruence
closure~\cite{nelson-oppen-78, kozen-stoc77, nelson}.
\Egraphs continue to be a critical component in successful SMT
solvers where they are used for combining satisfiability theories by sharing equality
information~\cite{z3}.
 A key difference between past implementations of \egraphs and
  \egg's \egraph is our novel rebuilding algorithm that maintains
  invariants only at certain critical points~(\autoref{sec:rebuild}).
This makes \egg more efficient for the purpose of equality saturation.
  \egg implements the pattern compilation strategy introduced by de Moura et al.~\cite{ematching}
  that is used in state of the art theorem provers~\cite{z3}.
Some provers~\cite{z3, simplify} propose optimizations like mod-time,
pattern-element and inverted-path-index to find new terms and relevant patterns
for matching, and avoid redundant matches. So far, we have found \egg to be
faster than several prior \egraph implementations even without
these optimizations. They are,
however, compatible with \egg's design and could be explored in the future. Another
key difference is \egg's powerful \eclass analysis abstraction and flexible
interface. They empower the programmer to easily leverage \egraphs for problems
involving complex semantic reasoning.

\paragraph{Congruence Closure}
Our rebuilding algorithm is similar to the congruence closure algorithm
  presented by \citet{downey-cse}.
The contribution of rebuilding is not \emph{how} it restores the \egraph invariants
  but \emph{when};
  it gives the client the ability to specialize invariant restoration to a
  particular workload like equality saturation.
Their algorithm also features a worklist of merges to be processed further,
  but it is offline, i.e.,
  the algorithm processes a given set of equalities and outputs the set of
  equalities closed over congruence.
Rebuilding is adapted to the online e-graph (and equality saturation) setting,
  where rewrites frequently examine the current set of equalities and assert new ones.
Rebuilding additionally propagates e-class analysis facts (\autoref{sec:analysis}).
Despite these differences,
  the core algorithms algorithms are similar enough that theoretical results on
  offline performance characteristics \cite{downey-cse} apply to both.
We do not provide theoretical analysis of rebuilding for the online setting;
  it is likely highly workload dependent.

\paragraph{Superoptimization and Equality Saturation}
The Denali~\cite{denali} superoptimizer first demonstrated how to use \egraphs
for optimized code generation as an alternative to hand-optimized machine code
and prior exhaustive approaches~\cite{massalin}, both of which were less
scalable. The inputs to Denali are programs in a C-like language from which it
produces assembly programs. Denali supported three types of
rewrites---arithmetic, architectural, and program-specific. After applying these
rewrites till saturation, it used architectural description of the hardware to
generate constraints that were solved using a SAT solver to output a
near-optimal program. While Denali's approach was a significant improvement over
prior work, it was intended to be used on straight line code only and therefore,
did not apply to large real programs.

Equality saturation~\cite{eqsat, eqsat-llvm} developed a compiler optimization
phase that works for complex language constructs like loops and conditionals.
The first equality saturation paper used an intermediate representation called
Program Expression Graphs (PEGs) to encode loops and conditionals. PEGs have
specialized nodes that can represent infinite sequences, which allows them to
represent loops. It uses a global profitability heuristic for extraction which
is implemented using a pseudo-boolean solver. Recently, \cite{yogo-pldi20} uses
PEGs for code search.
\egg can support PEGs as a user-defined language,
and thus their technique could be ported.


\section{Conclusion}
\label{sec:conclusion}

We presented two new techniques,
  rebuilding and \eclass analysis,
  that make equality saturation
  fast and extensible enough for a new family of
  applications.
Rebuilding is a new way to
  amortize the cost of maintaining the \egraph's
  data structure invariants,
  specializing the \egraph to the equality saturation workload.
\Eclass analysis is a general framework
  allowing for interpreted reasoning
  beyond what purely syntactic rewrites can provide.

We implemented both of these techniques in \egg,
  a reusable, extensible, and efficient \egraph library.
\Egg is generic over the user-defined language,
  which allowed focused effort on optimization and efficiency
  while obviating the need for
  ad hoc \egraph implementations and manipulations.
Our case studies show that equality saturation can now scale
  further and be used more flexibly than before;
  \egg provided new functionality and large speedups.
We believe that these contributions
  position equality saturation as a powerful toolkit for
  program synthesis and optimization.



\begin{acks}
  Thanks to our anonymous paper and artifact reviewers for their feedback.
  Special thanks to our shepherd Simon Peyton Jones,
    Leonardo de Moura,
    and many
    members of the \href{https://uwplse.org}{PLSE} group.
  This work was supported in part by
    the Applications Driving Architectures (ADA) Research Center,
    a JUMP Center co-sponsored by SRC and DARPA,
    as well as the National Science Foundation
    under Grant Nos. 1813166 and 1749570.
\end{acks}

\bibliography{references}


\begin{thebibliography}{33}


\ifx \showCODEN    \undefined \def \showCODEN     #1{\unskip}     \fi
\ifx \showDOI      \undefined \def \showDOI       #1{#1}\fi
\ifx \showISBNx    \undefined \def \showISBNx     #1{\unskip}     \fi
\ifx \showISBNxiii \undefined \def \showISBNxiii  #1{\unskip}     \fi
\ifx \showISSN     \undefined \def \showISSN      #1{\unskip}     \fi
\ifx \showLCCN     \undefined \def \showLCCN      #1{\unskip}     \fi
\ifx \shownote     \undefined \def \shownote      #1{#1}          \fi
\ifx \showarticletitle \undefined \def \showarticletitle #1{#1}   \fi
\ifx \showURL      \undefined \def \showURL       {\relax}        \fi
\providecommand\bibfield[2]{#2}
\providecommand\bibinfo[2]{#2}
\providecommand\natexlab[1]{#1}
\providecommand\showeprint[2][]{arXiv:#2}

\bibitem[\protect\citeauthoryear{Andries, Engels, Habel, Hoffmann, Kreowski,
  Kuske, Plump, Sch\"{u}rr, and Taentzer}{Andries et~al\mbox{.}}{1999}]%
        {graphs}
\bibfield{author}{\bibinfo{person}{Marc Andries}, \bibinfo{person}{Gregor
  Engels}, \bibinfo{person}{Annegret Habel}, \bibinfo{person}{Berthold
  Hoffmann}, \bibinfo{person}{Hans-J\"{o}rg Kreowski}, \bibinfo{person}{Sabine
  Kuske}, \bibinfo{person}{Detlef Plump}, \bibinfo{person}{Andy Sch\"{u}rr},
  {and} \bibinfo{person}{Gabriele Taentzer}.} \bibinfo{year}{1999}\natexlab{}.
\newblock \showarticletitle{Graph Transformation for Specification and
  Programming}.
\newblock \bibinfo{journal}{\emph{Sci. Comput. Program.}} \bibinfo{volume}{34},
  \bibinfo{number}{1} (\bibinfo{date}{April} \bibinfo{year}{1999}),
  \bibinfo{pages}{1–54}.
\newblock
\showISSN{0167-6423}
\urldef\tempurl%
\url{https://doi.org/10.1016/S0167-6423(98)00023-9}
\showDOI{\tempurl}


\bibitem[\protect\citeauthoryear{Boehm}{Boehm}{2019}]%
        {Boehm_2019}
\bibfield{author}{\bibinfo{person}{Matthias Boehm}.}
  \bibinfo{year}{2019}\natexlab{}.
\newblock \showarticletitle{Apache {SystemML}}.
\newblock \bibinfo{journal}{\emph{Encyclopedia of Big Data Technologies}}
  (\bibinfo{year}{2019}), \bibinfo{pages}{81–86}.
\newblock
\showISBNx{9783319775258}
\urldef\tempurl%
\url{https://doi.org/10.1007/978-3-319-77525-8_187}
\showDOI{\tempurl}


\bibitem[\protect\citeauthoryear{Boyle, Harmer, and Winter}{Boyle
  et~al\mbox{.}}{1996}]%
        {DBLP:conf/scitools/BoyleHW96}
\bibfield{author}{\bibinfo{person}{James~M. Boyle}, \bibinfo{person}{Terence~J.
  Harmer}, {and} \bibinfo{person}{Victor~L. Winter}.}
  \bibinfo{year}{1996}\natexlab{}.
\newblock \showarticletitle{The {TAMPR} Program Transformation System:
  Simplifying the Development of Numerical Software}. In
  \bibinfo{booktitle}{\emph{Modern Software Tools for Scientific Computing,
  SciTools 1996, Oslo, Norway, September 16-18, 1996}},
  \bibfield{editor}{\bibinfo{person}{Erlend Arge}, \bibinfo{person}{Are~Magnus
  Bruaset}, {and} \bibinfo{person}{Hans~Petter Langtangen}} (Eds.).
  \bibinfo{publisher}{Birkh{\"{a}}user}, \bibinfo{pages}{353--372}.
\newblock
\urldef\tempurl%
\url{https://doi.org/10.1007/978-1-4612-1986-6\_17}
\showDOI{\tempurl}


\bibitem[\protect\citeauthoryear{Davis and Putnam}{Davis and Putnam}{1960}]%
        {dpll}
\bibfield{author}{\bibinfo{person}{Martin Davis} {and} \bibinfo{person}{Hilary
  Putnam}.} \bibinfo{year}{1960}\natexlab{}.
\newblock \showarticletitle{A Computing Procedure for Quantification Theory}.
\newblock \bibinfo{journal}{\emph{J. ACM}} \bibinfo{volume}{7},
  \bibinfo{number}{3} (\bibinfo{date}{July} \bibinfo{year}{1960}),
  \bibinfo{pages}{201–215}.
\newblock
\showISSN{0004-5411}
\urldef\tempurl%
\url{https://doi.org/10.1145/321033.321034}
\showDOI{\tempurl}


\bibitem[\protect\citeauthoryear{de~Moura and Bj{\o}rner}{de~Moura and
  Bj{\o}rner}{2007}]%
        {ematching}
\bibfield{author}{\bibinfo{person}{Leonardo de Moura} {and}
  \bibinfo{person}{Nikolaj Bj{\o}rner}.} \bibinfo{year}{2007}\natexlab{}.
\newblock \showarticletitle{Efficient E-Matching for SMT Solvers}. In
  \bibinfo{booktitle}{\emph{Automated Deduction -- CADE-21}},
  \bibfield{editor}{\bibinfo{person}{Frank Pfenning}} (Ed.).
  \bibinfo{publisher}{Springer Berlin Heidelberg}, \bibinfo{address}{Berlin,
  Heidelberg}, \bibinfo{pages}{183--198}.
\newblock
\showISBNx{978-3-540-73595-3}


\bibitem[\protect\citeauthoryear{De~Moura and Bj{\o}rner}{De~Moura and
  Bj{\o}rner}{2008}]%
        {z3}
\bibfield{author}{\bibinfo{person}{Leonardo De~Moura} {and}
  \bibinfo{person}{Nikolaj Bj{\o}rner}.} \bibinfo{year}{2008}\natexlab{}.
\newblock \showarticletitle{Z3: An Efficient SMT Solver}. In
  \bibinfo{booktitle}{\emph{Proceedings of the Theory and Practice of Software,
  14th International Conference on Tools and Algorithms for the Construction
  and Analysis of Systems}} (Budapest, Hungary)
  \emph{(\bibinfo{series}{TACAS'08/ETAPS'08})}.
  \bibinfo{publisher}{Springer-Verlag}, \bibinfo{address}{Berlin, Heidelberg},
  \bibinfo{pages}{337--340}.
\newblock
\showISBNx{3-540-78799-2, 978-3-540-78799-0}
\urldef\tempurl%
\url{http://dl.acm.org/citation.cfm?id=1792734.1792766}
\showURL{%
\tempurl}


\bibitem[\protect\citeauthoryear{Dershowitz}{Dershowitz}{1993}]%
        {nachum-rewrites}
\bibfield{author}{\bibinfo{person}{Nachum Dershowitz}.}
  \bibinfo{year}{1993}\natexlab{}.
\newblock \bibinfo{booktitle}{\emph{A taste of rewrite systems}}.
\newblock \bibinfo{publisher}{Springer Berlin Heidelberg},
  \bibinfo{address}{Berlin, Heidelberg}, \bibinfo{pages}{199--228}.
\newblock
\showISBNx{978-3-540-47776-1}
\urldef\tempurl%
\url{https://doi.org/10.1007/3-540-56883-2_11}
\showDOI{\tempurl}


\bibitem[\protect\citeauthoryear{Dershowitz and Jouannaud}{Dershowitz and
  Jouannaud}{1990}]%
        {rewritesystems}
\bibfield{author}{\bibinfo{person}{Nachum Dershowitz} {and}
  \bibinfo{person}{Jean{-}Pierre Jouannaud}.} \bibinfo{year}{1990}\natexlab{}.
\newblock \showarticletitle{Rewrite Systems}.
\newblock In \bibinfo{booktitle}{\emph{Handbook of Theoretical Computer
  Science, Volume {B:} Formal Models and Semantics}},
  \bibfield{editor}{\bibinfo{person}{Jan van Leeuwen}} (Ed.).
  \bibinfo{publisher}{Elsevier and {MIT} Press}, \bibinfo{pages}{243--320}.
\newblock
\urldef\tempurl%
\url{https://doi.org/10.1016/b978-0-444-88074-1.50011-1}
\showDOI{\tempurl}


\bibitem[\protect\citeauthoryear{Detlefs, Nelson, and Saxe}{Detlefs
  et~al\mbox{.}}{2005}]%
        {simplify}
\bibfield{author}{\bibinfo{person}{David Detlefs}, \bibinfo{person}{Greg
  Nelson}, {and} \bibinfo{person}{James~B. Saxe}.}
  \bibinfo{year}{2005}\natexlab{}.
\newblock \showarticletitle{Simplify: A Theorem Prover for Program Checking}.
\newblock \bibinfo{journal}{\emph{J. ACM}} \bibinfo{volume}{52},
  \bibinfo{number}{3} (\bibinfo{date}{May} \bibinfo{year}{2005}),
  \bibinfo{pages}{365--473}.
\newblock
\showISSN{0004-5411}
\urldef\tempurl%
\url{https://doi.org/10.1145/1066100.1066102}
\showDOI{\tempurl}


\bibitem[\protect\citeauthoryear{Downey, Sethi, and Tarjan}{Downey
  et~al\mbox{.}}{1980}]%
        {downey-cse}
\bibfield{author}{\bibinfo{person}{Peter~J. Downey}, \bibinfo{person}{Ravi
  Sethi}, {and} \bibinfo{person}{Robert~Endre Tarjan}.}
  \bibinfo{year}{1980}\natexlab{}.
\newblock \showarticletitle{Variations on the Common Subexpression Problem}.
\newblock \bibinfo{journal}{\emph{J. ACM}} \bibinfo{volume}{27},
  \bibinfo{number}{4} (\bibinfo{date}{Oct.} \bibinfo{year}{1980}),
  \bibinfo{pages}{758–771}.
\newblock
\showISSN{0004-5411}
\urldef\tempurl%
\url{https://doi.org/10.1145/322217.322228}
\showDOI{\tempurl}


\bibitem[\protect\citeauthoryear{Du, Priya~Inala, Pu, Spielberg, Schulz, Rus,
  Solar-Lezama, and Matusik}{Du et~al\mbox{.}}{2018}]%
        {inverse}
\bibfield{author}{\bibinfo{person}{Tao Du}, \bibinfo{person}{Jeevana
  Priya~Inala}, \bibinfo{person}{Yewen Pu}, \bibinfo{person}{Andrew Spielberg},
  \bibinfo{person}{Adriana Schulz}, \bibinfo{person}{Daniela Rus},
  \bibinfo{person}{Armando Solar-Lezama}, {and} \bibinfo{person}{Wojciech
  Matusik}.} \bibinfo{year}{2018}\natexlab{}.
\newblock \showarticletitle{InverseCSG: automatic conversion of 3D models to
  CSG trees}. \bibinfo{pages}{1--16}.
\newblock
\urldef\tempurl%
\url{https://doi.org/10.1145/3272127.3275006}
\showDOI{\tempurl}


\bibitem[\protect\citeauthoryear{Ellis, Ritchie, Solar-Lezama, and
  Tenenbaum}{Ellis et~al\mbox{.}}{2018}]%
        {latex}
\bibfield{author}{\bibinfo{person}{Kevin Ellis}, \bibinfo{person}{Daniel
  Ritchie}, \bibinfo{person}{Armando Solar-Lezama}, {and}
  \bibinfo{person}{Joshua~B. Tenenbaum}.} \bibinfo{year}{2018}\natexlab{}.
\newblock \showarticletitle{Learning to Infer Graphics Programs from Hand-Drawn
  Images}. In \bibinfo{booktitle}{\emph{Neural Information Processing Systems
  (NIPS)}}.
\newblock


\bibitem[\protect\citeauthoryear{Jia, Padon, Thomas, Warszawski, Zaharia, and
  Aiken}{Jia et~al\mbox{.}}{2019}]%
        {taso}
\bibfield{author}{\bibinfo{person}{Zhihao Jia}, \bibinfo{person}{Oded Padon},
  \bibinfo{person}{James Thomas}, \bibinfo{person}{Todd Warszawski},
  \bibinfo{person}{Matei Zaharia}, {and} \bibinfo{person}{Alex Aiken}.}
  \bibinfo{year}{2019}\natexlab{}.
\newblock \showarticletitle{TASO: optimizing deep learning computation with
  automatic generation of graph substitutions}. In
  \bibinfo{booktitle}{\emph{Proceedings of the 27th ACM Symposium on Operating
  Systems Principles}}. \bibinfo{pages}{47--62}.
\newblock


\bibitem[\protect\citeauthoryear{Joshi, Nelson, and Randall}{Joshi
  et~al\mbox{.}}{2002}]%
        {denali}
\bibfield{author}{\bibinfo{person}{Rajeev Joshi}, \bibinfo{person}{Greg
  Nelson}, {and} \bibinfo{person}{Keith Randall}.}
  \bibinfo{year}{2002}\natexlab{}.
\newblock \showarticletitle{Denali: A Goal-directed Superoptimizer}.
\newblock \bibinfo{journal}{\emph{SIGPLAN Not.}} \bibinfo{volume}{37},
  \bibinfo{number}{5} (\bibinfo{date}{May} \bibinfo{year}{2002}),
  \bibinfo{pages}{304--314}.
\newblock
\showISSN{0362-1340}
\urldef\tempurl%
\url{https://doi.org/10.1145/543552.512566}
\showDOI{\tempurl}


\bibitem[\protect\citeauthoryear{Kozen}{Kozen}{1977}]%
        {kozen-stoc77}
\bibfield{author}{\bibinfo{person}{Dexter Kozen}.}
  \bibinfo{year}{1977}\natexlab{}.
\newblock \showarticletitle{Complexity of Finitely Presented Algebras}. In
  \bibinfo{booktitle}{\emph{Proceedings of the Ninth Annual ACM Symposium on
  Theory of Computing}} (Boulder, Colorado, USA) \emph{(\bibinfo{series}{STOC
  '77})}. \bibinfo{publisher}{Association for Computing Machinery},
  \bibinfo{address}{New York, NY, USA}, \bibinfo{pages}{164–177}.
\newblock
\showISBNx{9781450374095}
\urldef\tempurl%
\url{https://doi.org/10.1145/800105.803406}
\showDOI{\tempurl}


\bibitem[\protect\citeauthoryear{Massalin}{Massalin}{1987}]%
        {massalin}
\bibfield{author}{\bibinfo{person}{Henry Massalin}.}
  \bibinfo{year}{1987}\natexlab{}.
\newblock \showarticletitle{Superoptimizer: A Look at the Smallest Program}. In
  \bibinfo{booktitle}{\emph{Proceedings of the Second International Conference
  on Architectual Support for Programming Languages and Operating Systems}}
  (Palo Alto, California, USA) \emph{(\bibinfo{series}{ASPLOS II})}.
  \bibinfo{publisher}{IEEE Computer Society Press},
  \bibinfo{address}{Washington, DC, USA}, \bibinfo{pages}{122–126}.
\newblock
\showISBNx{0818608056}
\urldef\tempurl%
\url{https://doi.org/10.1145/36206.36194}
\showDOI{\tempurl}


\bibitem[\protect\citeauthoryear{Nandi, Wilcox, Panchekha, Blau, Grossman, and
  Tatlock}{Nandi et~al\mbox{.}}{2018}]%
        {reincarnate}
\bibfield{author}{\bibinfo{person}{Chandrakana Nandi},
  \bibinfo{person}{James~R. Wilcox}, \bibinfo{person}{Pavel Panchekha},
  \bibinfo{person}{Taylor Blau}, \bibinfo{person}{Dan Grossman}, {and}
  \bibinfo{person}{Zachary Tatlock}.} \bibinfo{year}{2018}\natexlab{}.
\newblock \showarticletitle{Functional Programming for Compiling and
  Decompiling Computer-aided Design}.
\newblock \bibinfo{journal}{\emph{Proc. ACM Program. Lang.}}
  \bibinfo{volume}{2}, \bibinfo{number}{ICFP}, Article \bibinfo{articleno}{99}
  (\bibinfo{date}{July} \bibinfo{year}{2018}), \bibinfo{numpages}{31}~pages.
\newblock
\showISSN{2475-1421}
\urldef\tempurl%
\url{https://doi.org/10.1145/3236794}
\showDOI{\tempurl}


\bibitem[\protect\citeauthoryear{Nandi, Willsey, Anderson, Wilcox, Darulova,
  Grossman, and Tatlock}{Nandi et~al\mbox{.}}{2020}]%
        {szalinski}
\bibfield{author}{\bibinfo{person}{Chandrakana Nandi}, \bibinfo{person}{Max
  Willsey}, \bibinfo{person}{Adam Anderson}, \bibinfo{person}{James~R. Wilcox},
  \bibinfo{person}{Eva Darulova}, \bibinfo{person}{Dan Grossman}, {and}
  \bibinfo{person}{Zachary Tatlock}.} \bibinfo{year}{2020}\natexlab{}.
\newblock \showarticletitle{Synthesizing Structured {CAD} Models with Equality
  Saturation and Inverse Transformations}. In
  \bibinfo{booktitle}{\emph{Proceedings of the 41st ACM SIGPLAN Conference on
  Programming Language Design and Implementation}} (London, UK)
  \emph{(\bibinfo{series}{PLDI 2020})}. \bibinfo{publisher}{Association for
  Computing Machinery}, \bibinfo{address}{New York, NY, USA},
  \bibinfo{pages}{31–44}.
\newblock
\showISBNx{9781450376136}
\urldef\tempurl%
\url{https://doi.org/10.1145/3385412.3386012}
\showDOI{\tempurl}


\bibitem[\protect\citeauthoryear{Nelson}{Nelson}{1980}]%
        {nelson}
\bibfield{author}{\bibinfo{person}{Charles~Gregory Nelson}.}
  \bibinfo{year}{1980}\natexlab{}.
\newblock \emph{\bibinfo{title}{Techniques for Program Verification}}.
\newblock \bibinfo{thesistype}{Ph.D. Dissertation}. \bibinfo{address}{Stanford,
  CA, USA}.
\newblock
\newblock
\shownote{AAI8011683.}


\bibitem[\protect\citeauthoryear{Nelson and Oppen}{Nelson and Oppen}{1980}]%
        {nelson-oppen-78}
\bibfield{author}{\bibinfo{person}{Greg Nelson} {and} \bibinfo{person}{Derek~C.
  Oppen}.} \bibinfo{year}{1980}\natexlab{}.
\newblock \showarticletitle{Fast Decision Procedures Based on Congruence
  Closure}.
\newblock \bibinfo{journal}{\emph{J. ACM}} \bibinfo{volume}{27},
  \bibinfo{number}{2} (\bibinfo{date}{April} \bibinfo{year}{1980}),
  \bibinfo{pages}{356–364}.
\newblock
\showISSN{0004-5411}
\urldef\tempurl%
\url{https://doi.org/10.1145/322186.322198}
\showDOI{\tempurl}


\bibitem[\protect\citeauthoryear{Nieuwenhuis and Oliveras}{Nieuwenhuis and
  Oliveras}{2005}]%
        {pp-congr}
\bibfield{author}{\bibinfo{person}{Robert Nieuwenhuis} {and}
  \bibinfo{person}{Albert Oliveras}.} \bibinfo{year}{2005}\natexlab{}.
\newblock \showarticletitle{Proof-Producing Congruence Closure}. In
  \bibinfo{booktitle}{\emph{Proceedings of the 16th International Conference on
  Term Rewriting and Applications}} (Nara, Japan)
  \emph{(\bibinfo{series}{RTA’05})}. \bibinfo{publisher}{Springer-Verlag},
  \bibinfo{address}{Berlin, Heidelberg}, \bibinfo{pages}{453–468}.
\newblock
\showISBNx{3540255966}
\urldef\tempurl%
\url{https://doi.org/10.1007/978-3-540-32033-3_33}
\showDOI{\tempurl}


\bibitem[\protect\citeauthoryear{Panchekha, Sanchez-Stern, Wilcox, and
  Tatlock}{Panchekha et~al\mbox{.}}{2015}]%
        {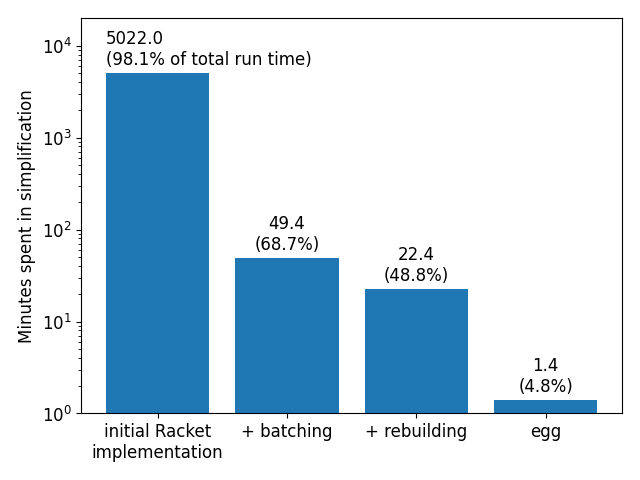}
\bibfield{author}{\bibinfo{person}{Pavel Panchekha}, \bibinfo{person}{Alex
  Sanchez-Stern}, \bibinfo{person}{James~R. Wilcox}, {and}
  \bibinfo{person}{Zachary Tatlock}.} \bibinfo{year}{2015}\natexlab{}.
\newblock \showarticletitle{Automatically Improving Accuracy for Floating Point
  Expressions}.
\newblock \bibinfo{journal}{\emph{SIGPLAN Not.}} \bibinfo{volume}{50},
  \bibinfo{number}{6} (\bibinfo{date}{June} \bibinfo{year}{2015}),
  \bibinfo{pages}{1–11}.
\newblock
\showISSN{0362-1340}
\urldef\tempurl%
\url{https://doi.org/10.1145/2813885.2737959}
\showDOI{\tempurl}


\bibitem[\protect\citeauthoryear{Premtoon, Koppel, and Solar-Lezama}{Premtoon
  et~al\mbox{.}}{2020}]%
        {yogo-pldi20}
\bibfield{author}{\bibinfo{person}{Varot Premtoon}, \bibinfo{person}{James
  Koppel}, {and} \bibinfo{person}{Armando Solar-Lezama}.}
  \bibinfo{year}{2020}\natexlab{}.
\newblock \showarticletitle{Semantic Code Search via Equational Reasoning}. In
  \bibinfo{booktitle}{\emph{Proceedings of the 41st ACM SIGPLAN Conference on
  Programming Language Design and Implementation}} (London, UK)
  \emph{(\bibinfo{series}{PLDI 2020})}. \bibinfo{publisher}{Association for
  Computing Machinery}, \bibinfo{address}{New York, NY, USA},
  \bibinfo{pages}{1066–1082}.
\newblock
\showISBNx{9781450376136}
\urldef\tempurl%
\url{https://doi.org/10.1145/3385412.3386001}
\showDOI{\tempurl}


\bibitem[\protect\citeauthoryear{Rust}{Rust}{[n.d.]}]%
        {rust}
\bibfield{author}{\bibinfo{person}{Rust}.} \bibinfo{year}{[n.d.]}\natexlab{}.
\newblock \bibinfo{booktitle}{\emph{Rust programming language}}.
\newblock
\urldef\tempurl%
\url{https://www.rust-lang.org/}
\showURL{%
\tempurl}


\bibitem[\protect\citeauthoryear{Sharma, Goyal, Liu, Kalogerakis, and
  Maji}{Sharma et~al\mbox{.}}{2017}]%
        {csgnet}
\bibfield{author}{\bibinfo{person}{Gopal Sharma}, \bibinfo{person}{Rishabh
  Goyal}, \bibinfo{person}{Difan Liu}, \bibinfo{person}{Evangelos Kalogerakis},
  {and} \bibinfo{person}{Subhransu Maji}.} \bibinfo{year}{2017}\natexlab{}.
\newblock \showarticletitle{CSGNet: Neural Shape Parser for Constructive Solid
  Geometry}.
\newblock \bibinfo{journal}{\emph{CoRR}}  \bibinfo{volume}{abs/1712.08290}
  (\bibinfo{year}{2017}).
\newblock
\showeprint[arxiv]{1712.08290}
\urldef\tempurl%
\url{http://arxiv.org/abs/1712.08290}
\showURL{%
\tempurl}


\bibitem[\protect\citeauthoryear{Stepp, Tate, and Lerner}{Stepp
  et~al\mbox{.}}{2011}]%
        {eqsat-llvm}
\bibfield{author}{\bibinfo{person}{Michael Stepp}, \bibinfo{person}{Ross Tate},
  {and} \bibinfo{person}{Sorin Lerner}.} \bibinfo{year}{2011}\natexlab{}.
\newblock \showarticletitle{Equality-Based Translation Validator for LLVM}. In
  \bibinfo{booktitle}{\emph{Computer Aided Verification}},
  \bibfield{editor}{\bibinfo{person}{Ganesh Gopalakrishnan} {and}
  \bibinfo{person}{Shaz Qadeer}} (Eds.). \bibinfo{publisher}{Springer Berlin
  Heidelberg}, \bibinfo{address}{Berlin, Heidelberg},
  \bibinfo{pages}{737--742}.
\newblock
\showISBNx{978-3-642-22110-1}


\bibitem[\protect\citeauthoryear{Tarjan}{Tarjan}{1975}]%
        {unionfind}
\bibfield{author}{\bibinfo{person}{Robert~Endre Tarjan}.}
  \bibinfo{year}{1975}\natexlab{}.
\newblock \showarticletitle{Efficiency of a Good But Not Linear Set Union
  Algorithm}.
\newblock \bibinfo{journal}{\emph{J. ACM}} \bibinfo{volume}{22},
  \bibinfo{number}{2} (\bibinfo{date}{April} \bibinfo{year}{1975}),
  \bibinfo{pages}{215–225}.
\newblock
\showISSN{0004-5411}
\urldef\tempurl%
\url{https://doi.org/10.1145/321879.321884}
\showDOI{\tempurl}


\bibitem[\protect\citeauthoryear{Tate, Stepp, Tatlock, and Lerner}{Tate
  et~al\mbox{.}}{2009}]%
        {eqsat}
\bibfield{author}{\bibinfo{person}{Ross Tate}, \bibinfo{person}{Michael Stepp},
  \bibinfo{person}{Zachary Tatlock}, {and} \bibinfo{person}{Sorin Lerner}.}
  \bibinfo{year}{2009}\natexlab{}.
\newblock \showarticletitle{Equality Saturation: A New Approach to
  Optimization}. In \bibinfo{booktitle}{\emph{Proceedings of the 36th Annual
  ACM SIGPLAN-SIGACT Symposium on Principles of Programming Languages}}
  (Savannah, GA, USA) \emph{(\bibinfo{series}{POPL '09})}.
  \bibinfo{publisher}{ACM}, \bibinfo{address}{New York, NY, USA},
  \bibinfo{pages}{264--276}.
\newblock
\showISBNx{978-1-60558-379-2}
\urldef\tempurl%
\url{https://doi.org/10.1145/1480881.1480915}
\showDOI{\tempurl}


\bibitem[\protect\citeauthoryear{Tian, Luo, Sun, Ellis, Freeman, Tenenbaum, and
  Wu}{Tian et~al\mbox{.}}{2019}]%
        {shape}
\bibfield{author}{\bibinfo{person}{Yonglong Tian}, \bibinfo{person}{Andrew
  Luo}, \bibinfo{person}{Xingyuan Sun}, \bibinfo{person}{Kevin Ellis},
  \bibinfo{person}{William~T. Freeman}, \bibinfo{person}{Joshua~B. Tenenbaum},
  {and} \bibinfo{person}{Jiajun Wu}.} \bibinfo{year}{2019}\natexlab{}.
\newblock \showarticletitle{Learning to Infer and Execute 3D Shape Programs}.
  In \bibinfo{booktitle}{\emph{International Conference on Learning
  Representations}}.
\newblock
\urldef\tempurl%
\url{https://openreview.net/forum?id=rylNH20qFQ}
\showURL{%
\tempurl}


\bibitem[\protect\citeauthoryear{van~den Brand, Heering, Klint, and
  Olivier}{van~den Brand et~al\mbox{.}}{2002}]%
        {DBLP:journals/toplas/BrandHKO02}
\bibfield{author}{\bibinfo{person}{Mark van~den Brand}, \bibinfo{person}{Jan
  Heering}, \bibinfo{person}{Paul Klint}, {and} \bibinfo{person}{Pieter~A.
  Olivier}.} \bibinfo{year}{2002}\natexlab{}.
\newblock \showarticletitle{Compiling language definitions: the {ASF+SDF}
  compiler}.
\newblock \bibinfo{journal}{\emph{{ACM} Trans. Program. Lang. Syst.}}
  \bibinfo{volume}{24}, \bibinfo{number}{4} (\bibinfo{year}{2002}),
  \bibinfo{pages}{334--368}.
\newblock
\urldef\tempurl%
\url{https://doi.org/10.1145/567097.567099}
\showDOI{\tempurl}


\bibitem[\protect\citeauthoryear{Visser, Benaissa, and Tolmach}{Visser
  et~al\mbox{.}}{1998}]%
        {DBLP:conf/icfp/VisserBT98}
\bibfield{author}{\bibinfo{person}{Eelco Visser},
  \bibinfo{person}{Zine{-}El{-}Abidine Benaissa}, {and}
  \bibinfo{person}{Andrew~P. Tolmach}.} \bibinfo{year}{1998}\natexlab{}.
\newblock \showarticletitle{Building Program Optimizers with Rewriting
  Strategies}. In \bibinfo{booktitle}{\emph{Proceedings of the third {ACM}
  {SIGPLAN} International Conference on Functional Programming {(ICFP} '98),
  Baltimore, Maryland, USA, September 27-29, 1998}},
  \bibfield{editor}{\bibinfo{person}{Matthias Felleisen}, \bibinfo{person}{Paul
  Hudak}, {and} \bibinfo{person}{Christian Queinnec}} (Eds.).
  \bibinfo{publisher}{{ACM}}, \bibinfo{pages}{13--26}.
\newblock
\urldef\tempurl%
\url{https://doi.org/10.1145/289423.289425}
\showDOI{\tempurl}


\bibitem[\protect\citeauthoryear{Wang, Hutchison, Leang, Howe, and Suciu}{Wang
  et~al\mbox{.}}{2020}]%
        {spores}
\bibfield{author}{\bibinfo{person}{Yisu~Remy Wang}, \bibinfo{person}{Shana
  Hutchison}, \bibinfo{person}{Jonathan Leang}, \bibinfo{person}{Bill Howe},
  {and} \bibinfo{person}{Dan Suciu}.} \bibinfo{year}{2020}\natexlab{}.
\newblock \showarticletitle{{SPORES}: Sum-Product Optimization via Relational
  Equality Saturation for Large Scale Linear Algebra}.
\newblock \bibinfo{journal}{\emph{Proceedings of the VLDB Endowment}}
  (\bibinfo{year}{2020}).
\newblock


\bibitem[\protect\citeauthoryear{Wu, Zhao, Nandi, Lipton, Tatlock, and
  Schulz}{Wu et~al\mbox{.}}{2019}]%
        {wu_siga19}
\bibfield{author}{\bibinfo{person}{Chenming Wu}, \bibinfo{person}{Haisen Zhao},
  \bibinfo{person}{Chandrakana Nandi}, \bibinfo{person}{Jeffrey~I. Lipton},
  \bibinfo{person}{Zachary Tatlock}, {and} \bibinfo{person}{Adriana Schulz}.}
  \bibinfo{year}{2019}\natexlab{}.
\newblock \showarticletitle{Carpentry Compiler}.
\newblock \bibinfo{journal}{\emph{ACM Transactions on Graphics}}
  \bibinfo{volume}{38}, \bibinfo{number}{6} (\bibinfo{year}{2019}),
  \bibinfo{pages}{Article No. 195}.
\newblock
\newblock
\shownote{presented at SIGGRAPH Asia 2019.}


\end{thebibliography}

\end{document}